\documentclass[11pt,reqno]{amsart}
\usepackage{amsaddr}
\usepackage[left=2.5cm,right=2.5cm,top=2cm,bottom=3cm]{geometry}
\usepackage[usenames,dvipsnames]{xcolor}
\usepackage{colortbl}
\usepackage[shortlabels]{enumitem}
\usepackage{amsmath,amsthm,amssymb,bm,bbm,dsfont}
\usepackage{thmtools}
\usepackage{thm-restate}

\usepackage{mathtools}\mathtoolsset{centercolon}
\mathtoolsset{showonlyrefs=true}
\usepackage[mathscr]{eucal}
\usepackage{upgreek}
\usepackage{setspace}
\usepackage{graphicx}
\usepackage{verbatim}
\usepackage{float}
\usepackage{placeins}
\usepackage{array}
\usepackage{booktabs}

\usepackage{threeparttable}
\usepackage[update,prepend]{epstopdf}
\usepackage{multirow}
\usepackage{amsfonts,amssymb,dsfont}
\usepackage[abs]{overpic}
\usepackage{xfrac}

\usepackage{tikz}
\usepackage{tikzscale}
\usetikzlibrary{calc}
\usetikzlibrary{patterns}
\usetikzlibrary{decorations.pathreplacing}
\usepackage{blox}

\bXLineStyle{-}

\definecolor{dullmagenta}{rgb}{0.4,0,0.4}   
\definecolor{darkblue}{rgb}{0,0,0.4}
\usepackage{hyperref}
\hypersetup{
	unicode=false,          
	pdftoolbar=true,        
	pdfmenubar=true,        
	pdffitwindow=false,     
	pdfstartview={FitH},    
	pdftitle={My title},    
	pdfauthor={Author},     
	pdfsubject={Subject},   
	pdfcreator={Creator},   
	pdfproducer={Producer}, 
	pdfkeywords={keyword1} {key2} {key3}, 
	pdfnewwindow=true,      
	linktocpage=true,
	colorlinks=true,        
	linkcolor=Red,          
	citecolor=ForestGreen,  
	filecolor=Magenta,      
	urlcolor=BlueViolet,    
}

\usepackage{doi}
\usepackage{caption, subcaption}
\usepackage{enumitem}
\usepackage{shadethm}

\makeatletter
\newcommand{\opnorm}{\@ifstar\@opnorms\@opnorm}
\newcommand{\@opnorms}[1]{%
	\left|\mkern-1.5mu\left|\mkern-1.5mu\left|
	#1
	\right|\mkern-1.5mu\right|\mkern-1.5mu\right|
}
\newcommand{\@opnorm}[2][]{%
	\mathopen{#1|\mkern-1.5mu#1|\mkern-1.5mu#1|}
	#2
	\mathclose{#1|\mkern-1.5mu#1|\mkern-1.5mu#1|}
}
\makeatother

\makeatletter
\ifx\@NODS\undefined%

\else%
\fi
\makeatother

\makeatletter
\ifx\@NODS\undefined%

\let\mathbb=\mathds
\else%
\fi
\makeatother

\usepackage{xargs}
\usepackage[prependcaption]{todonotes}
\newcommandx{\eric}[2][1=]{\todo[inline, author={Eric}, linecolor=yellow,backgroundcolor=yellow!25,bordercolor=yellow,#1]{#2}}

\newcommandx{\ericside}[2][1=]{\todo[author={Eric}, linecolor=yellow,backgroundcolor=yellow!25,bordercolor=yellow,#1]{#2}}

\DeclareMathOperator{\Tr}{Tr}
\DeclareMathOperator{\tr}{Tr}
\DeclareMathOperator{\e}{\mathrm{e}}

\newcommand{\DP}{D^{\mathbb{P}}}
\newcommand{\QP}{Q^{\mathbb{P}}}

\newcommand{\be}{{\mathbf e}}

\newcommand{\cH}{{\mathcal{H}}}
\newcommand{\cP}{{\mathcal{P}}}

\newcommand{\cX}{{\mathcal{X}}}

\newcommand{\cC}{{\mathcal{C}}}
\newcommand{\cR}{{\mathcal{R}}}

\newcommand{\cM}{{\mathcal{M}}}
\newcommand{\cE}{{\mathcal{E}}}
\newcommand{\cD}{{\mathcal{D}}}
\newcommand{\cB}{{\mathcal{B}}}
\newcommand{\cK}{{\mathcal{K}}}
\newcommand{\cN}{{\mathcal{N}}}
\newcommand{\cU}{{\mathcal{U}}}
\newcommand{\cL}{{\mathcal{L}}}

\def\0{{\mathbf{0}}}
\def\1{{\mathbf{1}}}
\def\2{{\mathbf{2}}}
\def\3{{\mathbf{3}}}
\def\4{{\mathbf{4}}}
\def\5{{\mathbf{5}}}
\def\6{{\mathbf{6}}}

\def\7{{\mathbf{7}}}
\def\8{{\mathbf{8}}}
\def\9{{\mathbf{9}}}

\def\be{\begin{equation}}
	\def\ee{\end{equation}}
\def\bea{\begin{eqnarray}}
	\def\eea{\end{eqnarray}}

\def\eps{\varepsilon}

\def\ep{\varepsilon}

\theoremstyle{plain}
\newtheorem{lemm}{Lemma} 
\newtheorem{theo}[lemm]{Theorem} 
\newtheorem{coro}[theo]{Corollary} 
\newshadetheorem{conj}{Conjecture}

\theoremstyle{definition}
\newtheorem{defn}{Definition}[section]

\theoremstyle{remark}
\newtheorem{remark}{Remark}[section]

\begin{document}
	
	\let\origmaketitle\maketitle
	\def\maketitle{
		\begingroup
		\def\uppercasenonmath##1{} 
		\let\MakeUppercase\relax 
		\origmaketitle
		\endgroup
	}
	
	\title{\bfseries \Large{ Strong Converse for Classical-Quantum\\ Degraded Broadcast Channels }}
	
	\author{ {Hao-Chung Cheng$^{1}$, Nilanjana Datta$^{1}$, Cambyse Rouz\'e$^{1,2}$ }}
	\address{\small  		
	$^{1}$Department of Applied Mathematics and Theoretical Physics, Centre for Mathematical Sciences\\University of Cambridge, Cambridge CB3 0WA, United Kingdom\\
	$^{2}$Technische Universit{\"a}t M{\"u}nchen, 80333 M{\"u}nchen, Germany}

	\email{\href{mailto:HaoChung.Ch@gmail.com}{HaoChung.Ch@gmail.com}, 
	\href{mailto:n.datta@statslab.cam.ac.uk}{n.datta@statslab.cam.ac.uk},
	\href{mailto:rouzecambyse@gmail.com}{rouzecambyse@gmail.com}}	
	
	\date{\today}

	\begin{abstract}
	We consider the transmission of classical information through a degraded broadcast channel, whose outputs are two quantum systems, with the state of one being a degraded version of the other. Yard \emph{et al.~}[\href{https://ieeexplore.ieee.org/document/6034754}{\textit{IEEE Trans.~Inf.~Theory}, \textbf{57}(10):7147--7162, 2011}] proved that the capacity region of such a channel is contained in a region characterized by certain entropic quantities. 
	We prove that this region satisfies the strong converse property, that is, the maximal probability of error incurred in transmitting information at rates lying outside this region converges to one exponentially in the number of uses of the channel. In establishing this result, we prove a second-order Fano-type inequality,
	which might be of independent interest. A powerful analytical tool which we employ in our proofs is the tensorization property of the quantum reverse hypercontractivity for the quantum depolarizing semigroup. 
	
	\end{abstract}
	
	\maketitle

\section{Introduction}

A broadcast channel models noisy one-to-many communication, examples of which abound in our daily lives. It can be used to transmit information to two\footnote{More generally, one can consider even more than two receivers.} receivers (say, Bob and Charlie) from a single sender (say, Alice). It was introduced by Cover in 1972 \cite{Cov72}. In the most general case, part of the information (the {\emph{common part}}) is intended for both the receivers, while part of the information ({\emph{the private part}}) consists of information intended for Bob and Charlie separately. Classically, the so-called discrete memoryless broadcast channel is modelled by a conditional probability distribution $\{p_{YZ|X}(y,z|x)\}$, where the random variables $X,Y,Z$ take values in ${\mathcal{X}}$ (the {\emph{input alphabet}}), and alphabets ${\mathcal{Y}}$ and ${\mathcal{Z}}$ (the {\emph{output alphabets}}), respectively. Hence, $X$ models Alice's input to the channel, while $Y$ and $Z$ correspond to the outputs received by Bob and Charlie, respectively. Suppose Alice sends her messages (or information) through multiple (say $n$) successive uses of such a channel, with $R_B$ and $R_C$ being the rates at which she transmits private information to Bob and Charlie, respectively, and $R$ being the rate at which she transmits common information to both of them. A triple $(R_B, R_C, R)$ is said to be an {\emph{achievable rate triple}} if the probability that an error is incurred in the transmission of the messages vanishes in the limit $n \to \infty$. In other words, these rates correspond to reliable transmission of information. Obviously, there is a tradeoff between these three rates: if one of them is high, the others are lowered in order to ensure that the common- as well as private information are transmitted reliably. The set of all achievable rate triples defines the {\emph{achievable rate region}}, and its closure defines the {\emph{capacity region}} of the broadcast channel. 

Determining the capacity region for a general broadcast channel remains a challenging open problem. However, certain special cases have been solved (see e.g.~\cite{Ber73, Wyn73, Ber74, Ber77, AGK76, Gal74, vdM75, Cov75, KM77, vdM77, Sat78, Mar79, Gam79, GvdM81, Cov98, Nai10, Ooh15a, Ooh15b, Ooh16, GK11}), the first of these being the case of the so-called {\emph{degraded broadcast channel}} (DBC). This is a broadcast channel for which the message that Charlie receives is a degraded version of the message that Bob receives. In other words, there exists a stochastic map which when acting on the message that Bob receives, yields the message that Charlie receives. Hence $p_{Y,Z|X}(y,z|x) = p_{Z|Y}(z|y)p_{Y|X}(y|x)$, and the three random variables $X,Y$ and $Z$ form a Markov Chain $X-Y-Z$.  Let us focus on the case in which there is no common information\footnote{The capacity region with common information can be obtained from the one without common information (see e.g. \cite[Chapter 5.7]{GK11}).} and hence the capacity region is specified by achievable rate pairs $(R_B, R_C)$. In this case, the capacity region has been shown to be given by~\cite{Ber73, Gal74, AK75}
\begin{align}
& \bigcup \{(R_B, R_C) \,: R_B \leq I(X;Y|U), R_C \leq I(U;Z)\},
\end{align}
where the union is over all joint probability distributions $\{p_{UX}(u,x)\}_{u \in \cU, x \in \cX}$, with $U$ being an auxiliary random variable taking values in an alphabet ${\mathcal{U}}$ with cardinality $|{\mathcal{U}}| \leq \min \{|{\mathcal{X}}|,|{\mathcal{Y}}|, |{\mathcal{Z}}|\} + 1.$ Here $I(X;Y|U)$ and $I(U;Z)$ denote the conditional mutual information (between $X$ and $Y$ conditioned on $U$) and the mutual information between $U$ and $Z$, respectively, and are the entropic quantities characterizing the achievable rate region.

In this paper, we consider a {\emph{classical-quantum degraded broadcast channel}} (c-q DBC), which we denote by $\mathscr{W}^{X \to BC}$. Here too, the input to the channel is classical and denoted by a random variable $X$ but the outputs are states of quantum systems $B$ and $C$. The channel is degraded in the sense that there exists some other quantum channel (say $\cN$), which when acting on the state of the system $B$ yields the state of the system $C$. Bob and Charlie receive the systems $B$ and $C$ respectively, and perform measurements on them in order to infer the classical messages that Alice sent to each of them. The channel is assumed to be memoryless and the achievable rates are computed in the asymptotic limit ($n \to \infty$, where $n$ denotes the number of successive uses of the channel). The achievable rate region for this channel was studied by Yard {\emph{et al.}}~\cite{YHD11} and later by Savov and Wilde~\cite{SW12}. Let $R_B$ and $R_C$ denote the rates at which Alice sends private information to Bob and Charlie respectively, and let $R$ be her rate of transmission of common information to both of them. 
See Figure~\ref{fig:protocol} for the illustration.

It was shown in \cite{YHD11} (see also \cite{SW12}) that any rate triple $(R, R_B, R_C) $ satisfying 
\begin{align}
\begin{split}
&R_B \leq I(X;B|U)_\sigma,\\
&R + R_C  \leq I(U;C)_\sigma,
\end{split}
\label{rf}
\end{align}
lies in the achievable rate region\footnote{For a precise definition of the achievable rate region and the capacity region, see Section \ref{sec:c-q_DBC}.}. Here the entropic quantities, appearing in the above inequalities are, taken with respect to a state $\sigma_{UXBC}$ of the following form
\begin{align}
\sigma_{UXBC} &= \sum_{(u,x)\in \mathcal{U}\times \mathcal{X} } p_U(u)\, p_{X|U}(x|u)\, |u \rangle \langle u|_U \otimes |x \rangle \langle x|_X \otimes \sigma^x_{BC} .
\end{align}
Here we use $X$ and $U$ to denote both random variables (taking values in finite sets $\cX$ and $\cU$, respectively), as well as quantum systems whose associated Hilbert spaces, $\cH_X$ and $\cH_U$, have complete orthonormal bases $\{|x\rangle\}$ and $\{|u\rangle\}$ labelled by the values taken by these random variables\footnote{Yard {\emph{et al.}}~\cite{YHD11} showed that it suffices to consider a random variable $U$ for which $|{\mathcal{U}}| \leq  \min\{|{\mathcal{X}}|, d_B^2 + d_C^2 -1 \}$. }. 
Hence, $|\mathcal{U}| := \dim \mathcal{H_U}$ and $|\mathcal{X}| := \dim \mathcal{H_X}$.

\smallskip

Moreover, Yard {\emph{et al.}} \cite[Theorem 2]{YHD11} established that the capacity region for such a c-q DBC is contained in a region specified by the following inequalities:
\begin{align}
\begin{split} \label{rfc}
&R_B  \leq I(X;B|U)_\omega,\\
&R + R_C  \leq I(U;C)_\omega,
\end{split}
\end{align}
for some state $\omega_{UXBC}$ of the following (more general) form, in which the system $U$ is a quantum system:
$$
\omega_{UXBC} = \sum_{x\in\mathcal{X} } p_X(x) \rho^x_{U}\otimes |x \rangle \langle x| \otimes \rho^x_{BC},
$$
where $\forall$ $x \in \cX$, $\rho^x_{U}$ is a state of the quantum system $U$.
\medskip

The above result establishes that for any rate triple $(R_B, R_C, R)$ which does not satisfy the inequalities \eqref{rf} for $\rho_{UXBC}$ of the above form, the {{maximum}} probability of incurring an error in the transmission of information is bounded away from zero, even in the asymptotic limit. In this paper, we show that the region spanned by such rate triples satisfies the so-called {\emph{strong converse property}}, that is, for any rate triple which lies outside this region, the {maximal} probability of error in the transmission of information is not only bounded away from zero but goes to one in the asymptotic limit. Moreover, the convergence to one is exponential in $n$.
A precise statement of this result is given by Corollary~\ref{coro:exponential} of Section \ref{main} below.
We first establish this strong converse property in the case in which no common information is sent (i.e.~$R=0$) and then discuss how this result can be extended to the general case in which both private and common information is sent by Alice. 

\begin{figure}[ht]
	\centering
	\includegraphics[width = 1\linewidth]{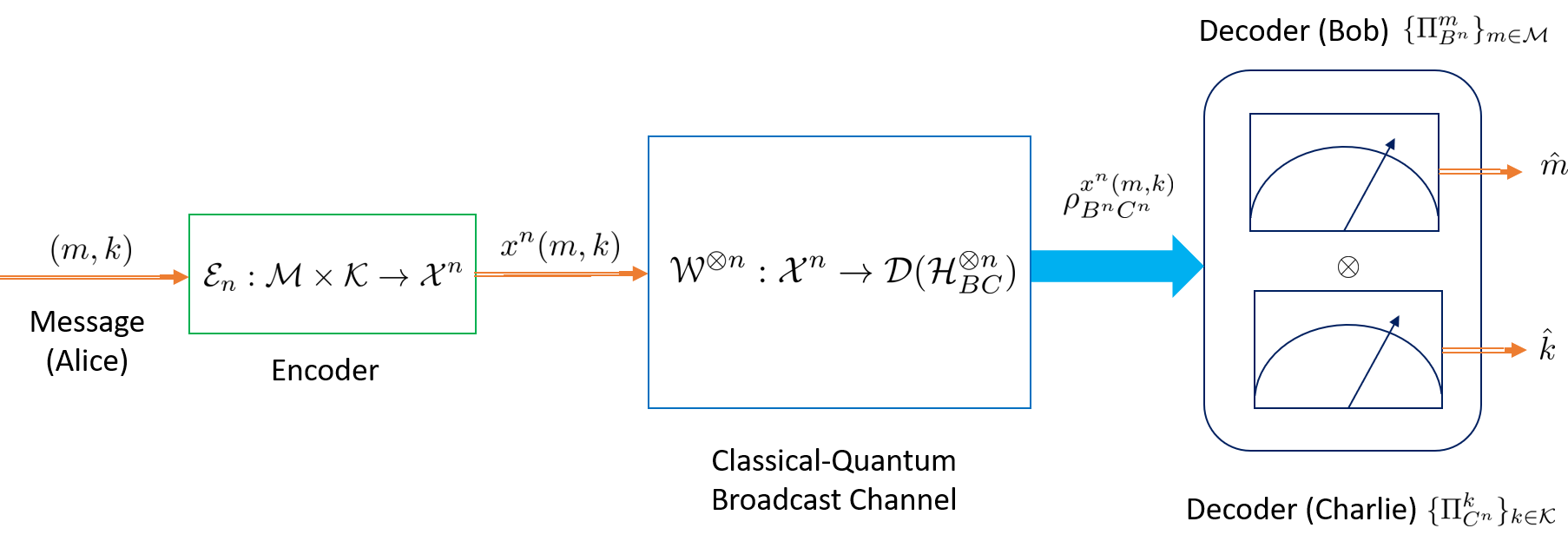}
	\caption{	The task of transmitting private information by Alice to Bob and Charlie through a classical-quantum broadcast channel. 
	We refer the readers to Section~\ref{sec:c-q_DBC} for detailed notation.}
		\label{fig:protocol}
\end{figure}

\textbf{Organization of the paper:} In Section~\ref{sec:notation}, we introduce necessary notation and the information-theoretic protocol of c-q DBC coding. 
In Section~\ref{main}, we state our main results.
In Section~\ref{sec:sc}, we prove a second-order Fano-type inequality for c-q channel coding, which is a main ingredient for establishing the second-order strong converse bound, which we prove in Section~\ref{sec:sc_DBC}.

\section{Notations and Definitions} \label{sec:notation}
Throughout this paper, we consider finite-dimensional Hilbert spaces, and discrete random variables which take values in finite sets. The subscript of a Hilbert space (say $B$), denotes the quantum system (say $\cH_B$) to which it is associated. We denote its dimension as $d_B := {\rm{dim}}\,{\mathcal{H}}_B$. 
Let $\mathbb{N}$, $\mathbb{R}$, and $\mathbb{R}_{\geq 0}$ be the set of natural numbers, real numbers, and non-negative real numbers, respectively.
Let ${\mathcal{B}}({\mathcal{H}})$ denote the algebra of linear operators acting on a Hilbert space ${\mathcal{H}}$, ${\mathcal{P}}({\mathcal{H}}) \subset  {\mathcal{B}}({\mathcal{H}})$ denote the set of positive semi-definite operators, ${\mathcal{D}}({\mathcal{H}}) \subset {\mathcal{P}}({\mathcal{H}})$ the set of quantum states (or density matrices): ${\mathcal{D}}({\mathcal{H}}) :\{ \rho \in {\mathcal{P}}({\mathcal{H}})\,:\, \tr [\rho] = 1\}$. A quantum operation (or quantum channel) is a superoperator given by a linear completely positive trace-preserving (CPTP) map. A quantum operation ${\mathcal{N}}^{A \to B}$ maps operators in ${\mathcal{B}}({\mathcal{H}}_A)$ to operators in ${\mathcal{B}}({\mathcal{H}}_B)$.  A superoperator $\Phi: {\mathcal{B}}({\mathcal{H}}) \to {\mathcal{B}}({\mathcal{H}})$ is said to be unital if $\Phi(\mathbb{I}) = \mathbb{I}$, where $\mathbb{I}$ denotes the identity operator in ${\mathcal{B}}({\mathcal{H}})$ .
We denote the identity superoperator as ${\rm{id}}$.
For any finite set $\mathcal{M}$, a positive-operator valued measure (POVM) on $\mathcal{H}$ is a set of positive semi-definite operators $\{\Pi^m\}_{m \in \mathcal{M}}$ satisfying $\Pi^m \in \mathcal{P(H)}$, for every $m\in\mathcal{M}$, and $\sum_{m\in\mathcal{M}} \Pi^m = \mathbb{I}$.

The von Neumann entropy of a state $\rho$ is defined as $S(\rho):= - \tr [\rho \log \rho]$, with the logarithm being taken to base $2$. 
The quantum relative entropy between a state $\rho \in {\mathcal{D}}({\mathcal{H}})$ and a positive semi-definite operator $\sigma$ is defined as
\begin{align}
D(\rho||\sigma) &: = \Tr\left[\rho (\log \rho - \log \sigma)\right].
\end{align}
It is well-defined if ${\rm{supp}}\,\rho \subseteq {\rm{supp}}\,\sigma$, and is equal to $+\infty$ otherwise. Here ${\rm{supp}}\,A$ denotes the support of the operator $A$. 
The quantum relative R\'enyi entropy of order $\alpha$,  for $\alpha \in (0,1)$, is defined as follows ~\cite{Pet86}:
\begin{align}
D_\alpha(\rho||\sigma) &:= \frac{1}{\alpha-1} \log \Tr\left[ \rho^\alpha \sigma^{1-\alpha} \right].
\end{align}
It is known that $D_\alpha(\rho||\sigma) \to D(\rho||\sigma)$ as $\alpha \to 1$  (see e.g.~\cite[Corollary 4.3]{Tom16}, \cite{Ume62}).
An important property satisfied by these relative entropies is the so-called {\em{data-processing inequality}}, which is given by
$ D_\alpha(\Lambda(\rho)\|\Lambda(\sigma)) \leq D_\alpha(\rho\|\sigma)$ for all $\alpha\in(0,1)$ and quantum operations $\Lambda$.
This induces corresponding data-processing inequalities for the quantities derived from these relative entropies, such as the quantum mutual information information \eqref{eq:mutual} and the conditional entropy \eqref{eq:conditional}.

For a bipartite state $\rho_{AB} \in \mathcal{D}(\mathcal{H}_A \otimes \mathcal{H}_B)$, the quantum mutual information and the conditional entropy are given in terms of the quantum relative entropy as follows:
\begin{align}
I(A;B)_\rho &= D\left( \rho_{AB} \| \rho_A \otimes \rho_B \right); \label{eq:mutual} \\
H(A|B)_\rho &= -D\left( \rho_{AB} \| \mathbb{I}_A \otimes \rho_B \right). \label{eq:conditional}
\end{align}
\medskip
The following inequality plays a fundamental role in our proofs.
\begin{lemm}
	[Araki-Lieb-Thirring inequality {\cite{Ari76, LT76}}] \label{lemm:ALT}
	For any $A,B\in \mathcal{P}(\mathcal{H})$, and $r\in[0,1]$,
	\begin{align}
	\Tr\left[ B^{\frac{r}{2}} A^r B^{\frac{r}{2}} \right] \leq \Tr\left[ \left( B^\frac12 A B^\frac12 \right)^r\, \right].
	\end{align}
\end{lemm}
\medskip

The proof of one of our main results (Theorem~\ref{theo:Fano}) employs a powerful analytical tool, namely, the so-called {\em{quantum reverse hypercontractivity}} of a certain quantum Markov semigroup (QMS) and its tensorization property. Let us introduce these concepts and the relevant results in brief. For more details see e.g.~\cite{BDR18} and references therein. The QMS that we consider is the so-called {\em{generalized quantum depolarizing semigroup (GQDS)}}. In the Heisenberg picture, for any state $\sigma >0$ on a Hilbert space $\cH$, the GQDS with invariant state $\sigma$ is defined by a one-parameter family of linear completely positive (CP) unital maps $\left(\Phi_t\right)_{t \geq 0}$, such that for any $X \in \cB(\cH)$,
\begin{align}\label{eq-gqds}
\Phi_t(X) = \e^{-t}X + (1-\e^{-t}) \tr[\sigma X] \,\mathbb{I}.
\end{align}
In the Schr\"odinger picture, the corresponding QMS is given by the family of CPTP maps $\left(\Phi^\star_t\right)_{t \geq 0}$, such that
$$ \tr[Y \Phi_t (X)] = \tr [\Phi^\star_t(Y) X], \quad \forall, \,\, X,Y \in \cB(\cH).$$
The action of $\Phi^\star_t$ on any state $\rho \in \cD(\cH)$ is that of a generalized depolarizing channel, which keeps the state unchanged
with probability $e^{-t}$, and replaces it by the state $\sigma$ with probability $(1 - \e^{-t})$:
$$ \Phi^\star_t(\rho) = \e^{-t}\rho + (1-\e^{-t}) \sigma\,.$$
 Note that $ \Phi^\star_t(\sigma) = \sigma$ for all $t\geq 0$, and that $\sigma$ is the unique invariant state of the evolution.
\medskip

To state the property of quantum reverse hypercontractivity, we define, for any $X\in\mathcal{B(H)}$, the non-commutative weighted $L_p$ norm {with respect to} the state $\sigma \in \mathcal{D}(\cH)$, for any $p \in \mathbb{R}\backslash\{0\}$\footnote{For $p<1$, these are pseudo-norms, since they do not satisfy the triangle inequality. For $p<0$, they are only defined for $X>0$ and for a non-full rank state by taking them equal to $\big( \Tr\big[ \big| \sigma^{-\frac{1}{2{p}}} X^{-1} \sigma^{-\frac{1}{2 {p}}} \big|^{{-p}} \big]\big)^{1/p}$.}:
\begin{align} \label{eq:wLp}
\left\| X \right\|_{p, \sigma} := \left( \Tr\left[ \left| \sigma^{\frac{1}{2{p}}} X \sigma^{\frac{1}{2 {p}}} \right|^{{p}} \right] \right)^{\frac{1}{{p}}}.
\end{align} 
\smallskip
A QMS $\left(\Phi_t\right)_{t \geq 0}$ is said to be reverse $p$-contractive for $p <1$, if
\begin{align}\label{contract}
|| \Phi_t(X)||_{p,\sigma} \geq ||X||_{p,\sigma}, \quad \forall \, X >0.
\end{align}
The GQDS can be shown to satisfy a stronger inequality: $\forall$ $ p<q <1$,
\begin{align}\label{QRHC}
  || \Phi_t(X)||_{p,\sigma} &\geq ||X||_{q,\sigma}, \quad \forall \, X >0,
\end{align}
for
\begin{align}\label{tcond}
t &\geq \frac{1}{4 \alpha_1 (\cL)} \log \left(\frac{p-1}{q-1}\right),
\end{align}
where $\alpha_1 (\cL)>0$ is the so called {\em{modified logarithmic Sobolev constant}}, and $\cL$ denotes
the generator of the GQDS, which is defined through the relation $\Phi_t(X) = \e^{-t\cL}(X)$ and is given by
$$\cL(X) = X - \tr[\sigma X] \mathbb{I}.$$
The inequality \eqref{QRHC} is indeed stronger than \eqref{contract} since the map $p \mapsto ||X||_{p,\sigma}$ is non-decreasing.
\medskip

In the context of this paper, instead of the GQDS defined through \eqref{eq-gqds}, we need to consider the QMS $\left(\Phi_{t, x^n}\right)_{t \geq 0}$,
with $\Phi_{t, x^n}$ being a CP unital map acting on $\cB(\cH^{\otimes n})$, and being labelled by sequences
$x^n \equiv (x_1,x_2, \ldots, x_n) \in \cX^n$, where $\cX$ is a finite set. For any $x \in \cX$, let $\rho^x \in \cD(\cH)$. Further, let
	\begin{align}
	  \rho^{x^n} &:= \rho^{x_1} \otimes \cdots \otimes \rho^{x_n} \,\in \cD(\cH^{\otimes n}).
	 \end{align}
Then,
  \begin{align}
    \Phi_{t,x^n} &:= \Phi_{t,x_1} \otimes \cdots \otimes \Phi_{t,x_n},
    \end{align}
    where $(\Phi_{t,x})_{t\ge 0}$ is a GQDS with invariant state $\rho^x$.
 We denote by $\cK_{x^n}=\sum_{i=1}^n\widehat{\cL}_{x_i}$ the generator of $(\Phi_{t,x^n})_{t\ge 0}$ where $\widehat{\cL}_{x_i}={\rm{id}}^{\otimes i-1} \otimes  {\cL}_{x_i} \otimes {\rm{id}}^{\otimes n-i}$, with ${\cL}_{x_i}$ being the generator of the GQDS $(\Phi_{t,x_i})_{t\geq 0}$. 
 If the modified logarithmic Sobolev constant $\alpha_1(\mathcal{K}_{x^n})$ is independent of $n$, or satisfies an $n$-independent lower bound, then it is called the \emph{tensorization property} of the GQDS. 
 The following tensorization property of the quantum reverse hypercontractivity of the above tensor product of GQDS was established in~\cite{BDR18} ( {See~\cite{mossel2013reverse} for its classical counterpart, as well as \cite{CKMT15} for its extension to doubly stochastic QMS}):
\begin{theo}[Quantum reverse hypercontractivity for tensor products of depolarizing semigroups {\cite[Corollary 17, Theorem 19]{BDR18}}] \label{lemm:RHC}
~\\
  For the QMS $\left(\Phi_{t, x^n}\right)_{t \geq 0}$ introduced above, for any $\mathsf{p}\leq \mathsf{q}< 1$ and for any $t$ satisfying $t\geq \log \frac{\mathsf{p}-1}{\mathsf{q}-1}$, the following inequality holds:
	\begin{align}
	\left\| \Phi_{t,x^n}(G_n) \right\|_{\mathsf{p},\rho^{x^n}}
	\geq \left\| G_n \right\|_{\mathsf{q},\rho^{x^n}}, \quad \forall\, G_n>0.
	\end{align}
	In other words, $\alpha_1(\cK_{x^n})\ge \frac{1}{4}$.
\end{theo}

\subsection{Classical-quantum (c-q) broadcast channel} \label{sec:c-q_DBC}
We define the \emph{classical-quantum (c-q) degraded broadcast channel} as follows.

\begin{defn} \label{defn:channel}
	A \emph{classical-quantum broadcast channel} ${\mathscr{W}}^{X\to BC}$ is a quantum operation defined as follows: 
	\begin{align}
	\begin{split}
	\mathscr{W} \equiv {\mathscr{W}}^{X\to BC}: \mathcal{X} &\to \mathcal{D}(\mathcal{H}_B\otimes \mathcal{H}_C); \\
	x &\mapsto \rho_{BC}^x.
	\end{split}
	\end{align}
	Here $X$ is a random variable which takes values in a finite set ${\mathcal{X}}$. A classical input $x \in {\mathcal{X}}$ to this channel, yields a quantum state $\rho_{BC}^x$ as output.
	Moreover, such a channel is said to be a c-q degraded broadcast channel (c-q DBC), if there exists a quantum channel ${\mathcal{N}}^{B \to C}$ such that $ \forall \, x \in {\mathcal{X}}$ the reduced state of the system $C$, $\rho^x_{C} = \tr_B (\rho^x_{BC}),$ satisfies 
	\begin{align}
	\rho^x_{C} &= {\mathcal{N}}^{B \to C}(\rho^x_B), \quad {\hbox{with}}\,\, \rho^x_B= \tr_C (\rho^x_{BC}).
	\end{align}
	Here $\tr_B$ and $\tr_C$ denote the partial traces over ${\mathcal{H}}_B$ and ${\mathcal{H}}_C$, respectively. As in the classical case, we consider Alice to be the sender (she hence holds $X$) while the quantum systems $B$ and $C$ are received by Bob and Charlie, respectively.
\end{defn}

As in the classical  case, we assume the channel to be \emph{memoryless} and consider multiple (say $n\in\mathbb{N}$) successive uses of it.
In this scenario, one hence considers a sequence of channels $\{\mathscr{W}^{\otimes n} \}_{n\in\mathbb{N}}$ such that for all $x^n \equiv (x_1,x_2,\ldots,x_n) \in \mathcal{X}^n$,
\begin{align}
\mathscr{W}^{\otimes n}(x^n) := \mathscr{W}(x_1) \otimes \mathscr{W}(x_2) \otimes\cdots \otimes \mathscr{W}(x_n).
\end{align}

As mentioned above, we first focus on the case in which there is no common information, and hence $R=0$. In this case, the inequalities \eqref{rfc} reduce to
\begin{align}
\begin{split} \label{rfcnc}
R_B & \leq I(X;B|U)_\omega,\\
R_C & \leq I(U;C)_\omega,
\end{split}
\end{align}
for some state
\begin{align} 
\omega_{UXBC} &= \sum_{x \in \mathcal{X} } p_X(x) \rho^x_{U}\otimes |x \rangle \langle x| \otimes \rho^x_{BC}.
\label{r2}
\end{align}


Let Alice's private messages to Bob and Charlie be labelled by the elements of the index sets $\cM := \{1,2,\ldots , |\cM|\}$ and $\cK := \{1,2,\ldots , |\cK|\}$, respectively. For any $(R_B, R_C) \in {\mathbb{R}}^2_{\geq 0}$, and any $n \in {\mathbb{N}}$, an $(n, R_B, R_C)$ code is given by the pair $\left(\mathcal{E}_n, \mathcal{D}_n\right)$,
where $\cE_n$ is the encoding map
\begin{align}
\begin{split}
\cE_{n} : \cM \times \cK &\to \cX^n;\\
 (m,k) &\mapsto x^n(m,k) = (x_1(m,k),x_2(m,k) \ldots,x_n(m,k)),
\end{split}
\end{align}
with $|\cM| = \lfloor 2^{nR_B}\rfloor$ and $|\cK| = \lfloor 2^{nR_C}\rfloor$.
Henceforth, for simplicity we assume that $2^{nR_B}$ and $2^{nR_C}$ are integers.
The decoding map $\cD_n$ consists of two POVMs: 
$\Pi_{B^n}:=\{\Pi_{B^n}^m\}_{m \in \cM}$, and $\Pi_{C^n}:=\{\Pi_{C^n}^k\}_{k \in \cK}$ where $\Pi_{B^n}^m \in {\cP(\cH_B^{\otimes n})}$ and $\Pi_{C^n}^k \in {\cP(\cH_C^{\otimes n})}$  for any $(m,k) \in \cM \times \cK$ and $\sum_{ m \in \cM } \Pi_{B^n}^m = \mathds{1}_{B^n}$ and $\sum_{ k \in \cK } \Pi_{C^n}^k = \mathds{1}_{C^n}$.

If the classical sequence $x^n(m,k)$ (which is the codeword corresponding to the message $(m,k)$) is sent through $n$ successive uses of the memoryless c-q DBC $\mathscr{W}^{X\to BC}$, the output is the product state
 \begin{align}
\rho_{B^n C^n}^{x^n(m,k)} &= \rho_{BC}^{x_1(m,k)}\otimes \rho_{BC}^{x_2(m,k)}\ldots \otimes \rho_{BC}^{x_n(m,k)} \in \mathcal{D}\left( \mathcal{H}_{BC}^{\otimes n} \right),
\end{align}
where $\mathcal{H}_{BC} \equiv \mathcal{H}_B \otimes \mathcal{H}_C$.
The probability that an error is incurred in sending the message $(m,k)$ is then given by
$$ 1-\Tr\left[ \rho_{B^n C^n}^{x^n(m,k)} \left(  \Pi_{B^n}^{m} \otimes \Pi_{C^n}^{k} \right) \right] .$$
The {\emph{maximal probability of error}} for the code $\left(\mathcal{E}_n, \mathcal{D}_n\right)$ is then defined as follows:
\begin{align}
p_{\max}\left(\mathcal{E}_n, \mathcal{D}_n\right) &:= \max_{(m,k) \in \cM \times \cK} 
\left( 1- \Tr\left[ \rho_{B^n C^n}^{x^n(m,k)} (\Pi_{B^n}^{m} \otimes \Pi_{C^n}^{k} )\right]\right).
\end{align}
and the {\emph{average probability of error}} for the code $\left(\mathcal{E}_n, \mathcal{D}_n\right)$ is defined as
\begin{align} \label{eq:avg}
p_{\text{avg}}\left(\mathcal{E}_n, \mathcal{D}_n\right) &:= \frac{1}{|\mathcal{M}||\mathcal{K}| } \sum_{(m,k) \in \cM \times \cK} 
\left( 1- \Tr\left[ \rho_{B^n C^n}^{x^n(m,k)} (\Pi_{B^n}^{m} \otimes \Pi_{C^n}^{k} )\right]\right).
\end{align}
For any $\eps \in [0,1]$, an $(n, R_B, R_C)$ code $\left(\mathcal{E}_n, \mathcal{D}_n\right)$ is said be an $(n, R_B, R_C, \eps)$ code if $p_{\max}\left(\mathcal{E}_n, \mathcal{D}_n\right) \leq \eps$. 
For a fixed $\eps \in [0,1)$, a rate pair $(R_B, R_C)$ is said to be \emph{$\eps$-achievable} (under the maximal error criterion) if there exists a sequence of $(n, R_B, R_C, \eps_n)$ codes such that $\eps_n \to \eps$ as $n \to \infty$.

A rate pair $(R_B, R_C)$ is {\emph{achievable}} if $\eps=0$. It is clear that any rate pair which is  achievable
is also $\eps$-achievable for all $\eps \in (0,1)$. 
For any $\eps \in [0, 1)$, let us define the \emph{$\eps$-achievable rate region} and the \emph{$\eps$-capacity region} of $\mathscr{W}$ as follows:
\begin{align} 
\begin{split} \label{eq:C_W_eps}
\cR_\mathscr{W}(\eps) &:= \left\{ (R_B, R_C) \in \mathbb{R}_{\geq 0}^2 \,:\text{ $(R_B,R_C)$ is $\eps$-achievable} \right\}; \\
\cC_\mathscr{W}(\eps) &:= \overline{\cR_\mathscr{W}(\eps)},
\end{split}
\end{align}
where $\overline{ \mathcal{R}_\mathscr{W} } (\eps)$ denotes the closure of the set ${ \mathcal{R}_\mathscr{W}  }(\eps)$.
The \emph{capacity region} of $\mathscr{W}$ is then $\cC_\mathscr{W}(0)$.  It is clear that
\begin{align}
\cR_\mathscr{W}(0) &= \bigcap_{\eps \in (0,1)} \cR_\mathscr{W}(\eps); \quad \cC_\mathscr{W}(0) = \bigcap_{\eps \in (0,1)} \cC_\mathscr{W}(\eps).
\end{align}
Similarly, one can introduce the $\ep$-capacity region under the average error criterion, which we denote as $\cC_{\mathscr{W},\operatorname{avg}}(\eps)$.
Since the average probability of error of a code is always less than or equal to the associated maximal probability of error, the inclusion $\cC_{\mathscr{W}}(\eps) \subseteq \cC_{\mathscr{W}, \text{avg}}(\eps)$ holds for all $\eps \in[0,1)$.
Furthermore, a standard \emph{codebook expurgation method} \cite{Wil90}, \cite[Problem~8.11]{GK11} shows that it is possible to construct a sequence of $(n, R_B - \frac2n \log n, R_C - \frac2n \log n)$ code with maximal probability of error less than $\sqrt{\eps_n}$ if a sequence $(n, R_B, R_C)$ code with average probability of error $\eps_n$ exists such that $\eps_n \to 0$ as $n\to \infty$.
Hence, $$ \cC_{\mathscr{W},\text{ave}}(0) = \cC_{\mathscr{W}}(0).$$
For convenience, we will only focus on the $\eps$-capacity region $\cC_{\mathscr{W}}(\eps)$ under the maximal error criterion throughout this paper.

Let us define the following entropic regions 
\begin{align}
\begin{split}
\mathcal{R}_\mathscr{W}^{\text{ent}} &:= \bigcup  \left\{ (R_B, R_C) \in \mathbb{R}_{\geq 0}^2 \, :  R_B  \leq I(X;B|U)_\omega, R_C  \leq I(U;C)_\omega \right\}; \\
\mathcal{C}_\mathscr{W}^{\text{ent}} &:= \overline{ \mathcal{R}_\mathscr{W} }, 
\end{split} \label{eq:C_W}
\end{align}
where the union is taken over all states $\omega_{UXBC}$ of the form \eqref{r2}. 
Yard {\emph{et al.}}~showed that \cite[Theorem 2]{YHD11} 
\begin{align} \label{eq:Yard}
\cC_\mathscr{W}(0) \subseteq \cC_\mathscr{W}^{\text{ent}}.
\end{align}
We now have all the definitions needed to state our main results.

\section{Main Results}\label{main}
For the memoryless c-q DBC $\mathscr{W}\equiv \mathscr{W}^{X\to BC}$ defined above, the results that we obtain can be briefly summarized as follows. For  more detailed and precise statements of these results, see the relevant corollaries and theorems given in Section~\ref{sec:sc}.
\smallskip

\noindent
\emph{\bf{Result 1 [Strong converse property, Corollary~\ref{coro:epsilon}]}} For any $\eps\in (0,1)$
\begin{align} \cC_\mathscr{W} (\eps) &\subseteq \cC_\mathscr{W}^{\text{ent}},
\end{align} 
where $\cC_\mathscr{W} (\eps)$ denotes its $\eps$-capacity region (defined in \eqref{eq:C_W_eps}),
whereas $ \cC_\mathscr{W}^{\text{ent}}$ is the region characterized by entropic quantities given in \eqref{eq:C_W}.

\medskip
 
This result implies that for any sequence of $(n, R_B, R_C)$ codes $\left(\mathcal{E}_n, \mathcal{D}_n\right)$, for which the rate pair $(R_B, R_C)$ lies outside the region $ \cC_\mathscr{W}^{\text{ent}}$ of the c-q DBC $\mathscr{W}$, 
\begin{align}
p_{\max} \left(\mathcal{E}_n, \mathcal{D}_n\right) &\to 1 \quad {\hbox{as}} \; n \to \infty.\label{conv}
\end{align}
This establishes the strong converse property of $\cC_\mathscr{W}^{\operatorname{ent}}$.

\medskip
\noindent
\emph{\bf{Result 2 [Exponential convergence, Corollary~\ref{coro:exponential}]}} The convergence in \eqref{conv} is exponential in $n$:
\begin{align}
p_{\max} \left(\mathcal{E}_n, \mathcal{D}_n\right) & \geq 1 - \e^{-nf}, \quad \forall n \in \mathbb{N},
\end{align}
where  $f = ( \sqrt{ (\sqrt{d_B} + \sqrt{d_C})^2 + \eta } - \sqrt{d_B} - \sqrt{d_C} )^2 > 0$
for some $\eta >0$, which depends only on how far the rate pair $(R_B, R_C)$ is from the region $\mathcal{C}_\mathscr{W}^{\operatorname{ent}}$.
\medskip

{\bf{Proof Ingredients:}} We prove the above results by first strengthening \eqref{eq:Yard} \cite[Theorem 2]{YHD11} by establishing second order (in $n$) upper bounds on $\eps$-achievable rate pairs $(R_B,R_C)$; see Theorem~\ref{theo:SC_DBC} of Section~\ref{sec:sc}. The key ingredient of the proof of this result is a \emph{second-order Fano-type inequality for c-q channel coding} (Theorem~\ref{theo:Fano}), which we consider to be a result of independent interest. The latter in turn employs the powerful analytical tool described in Theorem \ref{lemm:RHC}, namely the tensorization property of the {\emph{quantum reverse hypercontractivity}} for the quantum depolarizing semigroup \cite{BDR18, LHV18, CDR19a}.

\section{Second-Order Fano-type inequality} \label{sec:sc}
In this section we give precise statements of our results (which were summarized in Section \ref{main}) and their proofs.
In Theorem~\ref{theo:Fano} below, we establish a second-order Fano-type inequality for standard classical-quantum (c-q) channel\footnote{That is, for a point-to-point channel with a single user and a single receiver, as opposed to a broadcast c-q channel.} coding. This theorem is a key ingredient in the proof of the second-order strong converse bound for the c-q degraded broadcast channel $\mathscr{W}^{X \to BC}$ (Theorem~\ref{theo:SC_DBC}), which leads to our main results (Corollaries \ref{coro:epsilon} and \ref{coro:exponential}).

\begin{theo}
	[Second-order Fano-type inequality for c-q channel coding] \label{theo:Fano}
	Let $\mathcal{M},\mathcal{K}$, and $\mathcal{X}$ denote arbitrary finite sets, and let the map $x\mapsto \rho_B^x \in \mathcal{D}(\mathcal{H}_B)$ denote a c-q channel for all $x\in\mathcal{X}$.
	Consider the following encoding map: $\forall$ $m\in\mathcal{M}$,
  $$\cE_n: m \mapsto x^n(m,k) \in \mathcal{X}^n\quad {\hbox{with probability}}\,\, q(k),$$ 
where $\{q(k)\}_{k \in \cK}$ denotes an arbitrary probability distribution on $\mathcal{K}$.
	Further, let $\cD_n$ denote a decoding map given by a POVM $\{\Pi^m_{B_n}\}_{m \in \mathcal{M}}$. If $(\cE_n,\cD_n)$
are such that for some $\eps \in (0,1)$,
	\begin{align} \label{eq:Fano_condition}
	\prod_{(m,k) \in \mathcal{M}\times \mathcal{K}}  \left(  \Tr\left[ \rho_{B^n}^{x^n(m,k)} \Pi_{B^n}^m \right] \right)^{\frac{1}{|\mathcal{M}|}q(k)} \geq 1-\eps,
	\end{align}
	then
	\begin{align}\label{ineq1}
	\log |\mathcal{M}| \leq I(M;B^n)_\rho + 2 \sqrt{ n d_B  \log \frac{1}{1-\eps} } + \log\frac{1}{1-\eps}.
	\end{align}
	In the above, the mutual information  is taken with respect to a state $\rho_{MB^n}$ which is the reduced state of
$$\rho_{MKB^n} := \frac{1}{|\mathcal{M}|} \sum_{(m,k)\in\mathcal{M}\times\mathcal{K}} q(k) |m\rangle \langle m| \otimes |k\rangle \langle k| \otimes \rho_{B^n}^{x^n(m,k)},$$  
with $\rho_{B^n}^{x^n(m,k)}= \bigotimes_{i=1}^n\rho_B^{x_i(m,k)}$ being an $n$-fold product state on $\mathcal{H}_B^{\otimes n}$.
\end{theo}
\smallskip

\begin{remark}
	We refer to it as a Fano-type inequality because of the following. The usual (classical) Fano inequality \cite{Fan61} can be cast in the following form:
	Let $M, \widehat{M}$ denote two random variables taking values in the same finite set $\mathcal{M}$. If $\Pr(M\neq \widehat{M}) = \eps \in [0,1)$.
	Then, the (classical) Fano inequality \cite{Fan61} states that
	\begin{align}\label{Fanoclass}
	H(M) \leq I(M;\widehat{M}) + h(\eps) + \eps\log \left(|\mathcal{M}| - 1 \right),
	\end{align}
	where $h(\eps) := -\eps \log \eps - (1-\eps) \log (1-\eps)$ is the binary entropy function.
	
	In Theorem~\ref{theo:Fano}, the random variable $M$ is equiprobable and hence $H(M) = \log |\mathcal{M}|$.
	Considering $\widehat{M}$ to be the random variable denoting the outcome of the POVM $\{\Pi^m_{B_n}\}_{m \in \mathcal{M}}$ on the state $\rho_{B^n}^{x^n(m,k)}$, and using the data-processing inequality for the mutual information, one can upper bound the right-hand side of \eqref{Fanoclass} by
	\begin{align}
	\log |\mathcal{M}| &\leq I(M;B^n) + h(\eps) + \eps\log \left(|\mathcal{M}| - 1 \right)\\
	&\le I(M;B^n) + h(\eps) + \eps \log |\mathcal{M}|
	\end{align}
	which can be rewritten as
	\begin{align}  \label{eq:usual_Fano}
	\log |\mathcal{M}| \leq \frac{1}{1-\eps} I(M;B^n)_\rho + f(\eps),
	\end{align}
	where $f(\eps) = \frac{h(\eps) }{1-\eps}$.
	The similarity between \eqref{eq:usual_Fano} and \eqref{ineq1} lead Liu \emph{et al.}~\cite{LHV18} to refer to the latter as a Fano-type inequality in the classical case.
	The phrase `second-order' is used because the right-hand side of \eqref{ineq1} explicitly gives a term of order $\sqrt{n}$.
\end{remark}

\begin{remark} 
The above theorem is a generalization of Theorem 32 of \cite{BDR18}, in which an inequality similar to \eqref{ineq1} was obtained\footnote{A classical analogue of Theorem 32 of \cite{BDR18} was earlier proved in~\cite{LHV18}.}. The main difference between the two is that 
in \cite{BDR18}, the mutual information, arising in the inequality, was evaluated with respect to a state which is a direct sum of tensor product states. In contrast, in Theorem \ref{theo:Fano}, the mutual
information is with respect to states which have a more general form, namely, they are direct sums of mixtures of tensor product states (i.e.~separable states):
\begin{align}
\rho_{MB^n} =\frac{1}{|\mathcal{M}|} \sum_{m\in\mathcal{M}}  |m\rangle \langle m|  \otimes \rho^m_{B^n}, 
\end{align}
where
	\begin{align} \label{eq:Fano17}
	\rho^m_{B^n} := \sum_{k\in\mathcal{K}} q(k) \rho_{B^n}^{x^n(m,k)}, \quad \forall\, m\in\mathcal{M},
	\end{align}
where $ \rho_{B^n}^{x^n(m,k)}= \bigotimes_{i=1}^n\rho_B^{x_i(m,k)}$.
This generalization is crucial for our proof of the strong converse property of a c-q DBC.
\end{remark}

\begin{remark} \label{remark:error}
	The condition given by the inequality \eqref{eq:Fano_condition} is called the \emph{geometric average error criterion}. It is stronger than the \emph{average error criterion}, 
$$\frac{1}{|\mathcal{M}|} \sum_{(m, k) \in \mathcal{M}\times \mathcal{K} } q(k) \Tr[ \rho_{B^n}^{x^n(m,k)} \Pi_{B^n}^m ] \geq 1-\eps,$$ in the classical Fano inequality \cite{Fan61}, but is weaker than the \emph{maximal error criterion}, 
$$ \min_{(m, k) \in \mathcal{M}\times \mathcal{K} } \Tr[ \rho_{B^n}^{x^n(m,k)} \Pi_{B^n}^m ] \geq 1-\eps.$$
	Since the Fano-type inequality is a tool to prove converse results in network information theory, one might wonder if it still holds under a weaker error criterion.
	 In the classical case, Liu \emph{et al.}~showed that an analogous  second-order Fano-type inequality does not hold if the geometric average error criterion is replaced by the average error criterion~\cite{LCV17}, \cite[Remark 3.3]{LHV18}.
	However, by a standard technique known as the \emph{codebook expurgation} (see e.g.~\cite[Problem 8.11]{GK11}, \cite{Wil90}, \cite{CK11}, \cite{Ahl78}), which consists of discarding codewords corresponding to large error probabilities, one might still be able to show a second-order converse bound under the average error criterion in certain network information-theoretic tasks.
\end{remark}

\begin{proof}
	[Proof of Theorem~\ref{theo:Fano}]
	Before starting the proof, we introduce necessary definitions that will be used later. Consider the QMS $\left(\Phi_{t,x^n(m,k)}\right)_{t\geq 0}$, where for all $(m, k) \in \cM \times \cK$
	\begin{align}
	&\Phi_{t,x^n(m,k)} := \Phi_{t,x_1(m,k)} \otimes \cdots \otimes \Phi_{t,x_n(m,k)}, \label{eq:semigroup}
\end{align}
and $\forall$ $x(m,k) \in \cX$, $\Phi_{t,x(m,k)}$ denotes the superoperator defining the GQDS \eqref{eq-gqds}: 
\begin{align}
\Phi_{t,x(m,k)}( T) &:= \e^{-t} T + (1-\e^{-t}) \Tr\left[ \rho_B^{x(m,k)} T \right] \mathbb{I}_B, \quad \forall T \in\mathcal{B}(\mathcal{H}_B),\, t>0
\end{align}
Further, we define the following superoperator
\begin{align}
&\Psi_{t}(T) := \e^{-t} T + (1-\e^{-t}) \Tr\left[ T \right] \mathbb{I}_B, \quad \forall\, T \in\mathcal{B}(\mathcal{H}_B).
 \end{align}
	For any $\rho, \sigma \in\mathcal{D}(\mathcal{H}_B)$, the \emph{projectively measured R\'enyi relative entropy} is defined as \cite{Don86, HP91, Pet86b}:
	\begin{align}
	\DP_\alpha(\rho\|\sigma) &:= 
	\frac{1}{\alpha-1} \log \QP_\alpha(\rho\|\sigma), \quad \forall \alpha \in (0,1),  \\
	{\hbox{with}}\quad	\QP_\alpha(\rho\|\sigma) &:= \inf_{ \{P_i\}_{i=1}^{d_B  }} \left\{ \sum_{i=1}^{d_B }  \left(\Tr [ P_i \rho ] \right)^\alpha \left(\Tr [ P_i \sigma ] \right)^{1-\alpha}	\right\}, \quad \forall \alpha \in (0,1), \label{eq:QP}
	\end{align}
	where the optimization is over all sets of mutually orthogonal projectors $ \{P_i\}_{i=1}^{d_B } $ on $\mathcal{H}_B$. 
	
	Let $m \in \mathcal{M}$, $t>0$, $p\in(0,\sfrac12)$ and let $\hat{p} = (1-\sfrac1p)^{-1} \in (-1,0)$ be its H\"older conjugate. We commence the proof by invoking a variational formula for $\QP_p$ \cite[Lemma 3]{BFT17}:
	\begin{align}
	\QP_p( \rho_{B^n} \| \rho_{B^n}^m ) &= \inf_{G_n>0} \left\{ \left( \Tr\left[ \rho_{B^n} G_n \right] \right)^p \left( \Tr\left[ \rho_{B^n}^m G_n^{ \hat p}  \right]\right)^{1-p} \right\} \\
	&\leq  \left( \Tr\left[ \rho_{B^n} \Psi_t^{\otimes n}(\Pi_{B^n}^m) \right] \right)^p \left( \Tr\left[ \rho_{B^n}^m (\Psi_t^{\otimes n}(\Pi_{B^n}^m))^{ \hat p}  \right]\right)^{1-p} \\
	&=  \left( \Tr\left[ \rho_{B^n} \Psi_t^{\otimes n}(\Pi_{B^n}^m) \right] \right)^p \left(  \sum_{k\in\mathcal{K}} q(k) \Tr\left[ \rho_{B^n}^{x^n(m,k)} (\Psi_t^{\otimes n}(\Pi_{B^n}^m))^{ \hat p}  \right]\right)^{1-p}. \label{eq:Fano1}
	\end{align}	
	In the above, we remark that $\Psi_t^{\otimes n}(\Pi_{B^n}^m) >0$ for all $t>0$ due to the definition of $\Psi_t$ and the condition \eqref{eq:Fano_condition}.
	
	Applying the Araki-Lieb-Thirring inequality, Lemma~\ref{lemm:ALT}, with $r= -\hat{p} \in(0,1)$, $A = (\Psi_t^{\otimes n}(\Pi_{B^n}^m))^{-1} > 0$, 
	and $B^{r} = \rho_{B^n}^{x^n(m,k)} $ yields
	\begin{align}
	\Tr\left[ \rho_{B^n}^{x^n(m,k)} (\Psi_t^{\otimes n}(\Pi_{B^n}^m))^{\hat{p}} \right] &\leq \Tr\left[ \left( \left(\rho_{B^n}^{x^n(m,k)} \right)^{- \frac{1}{2\hat{p}} } \left( \Psi_t^{\otimes n}(\Pi_{B^n}^m) \right)^{-1} \left(\rho_{B^n}^{x^n(m,k)} \right)^{- \frac{1}{2\hat{p}} } \right)^{-\hat{p}} \right] \\
	&= \left\| \Psi_t^{\otimes n}(\Pi_{B^n}^m) \right\|_{\hat{p}, \rho_{B^n}^{x^n(m,k)} }^{\hat{p}}. \label{eq:Fano2}
	\end{align}
	On the other hand, it is clear from the the definition \eqref{eq:QP} that $\QP_p( \rho_{B^n} \| \rho_{B^n}^m ) = \QP_{1-p} (\rho_{B^n}^m \| \rho_{B^n})$.
	Combining \eqref{eq:Fano1} and \eqref{eq:Fano2}, taking logarithms of both sides of the resulting inequality, and dividing by $p$, yields
	\begin{align}
	\DP_{1-p}  \left(\rho_{B^n}^m \| \rho_{B^n}\right) &\geq \frac{1}{\hat{p}} \log \left(   \sum_{k\in\mathcal{K}} q(k) \left\| \Psi_{t}^{\otimes n}( \Pi_{B^n}^m ) \right\|_{\hat{p}, \rho_{B^n}^{x^n(m,k)} }^{\hat{p}} \right) - \log \Tr\left[ \rho_{B^n} \Psi_{t}^{\otimes n}( \Pi_{B^n}^m ) \right].  \label{eq:Fano16}
	\end{align}
	Further, the left-hand side of \eqref{eq:Fano16} can be upper bounded using the data processing inequality for the \emph{relative R\'enyi entropy} with respect to projective measurements, i.e.~for $p\in(0,\sfrac12)$,
	\begin{align}
	\DP_{1-p}  \left(\rho_{B^n}^m \| \rho_{B^n}\right) \leq D_{1-p}  \left(\rho_{B^n}^m \| \rho_{B^n}\right) := \frac{1}{-p} \log \Tr\left[ (\rho_{B^n}^m)^{1-p} (\rho_{B^n})^p \right].
	\end{align}
	Averaging over all $m \in \mathcal{M}$, we have
	\begin{align}
	\begin{split}
	\frac{1}{|\mathcal{M}|} \sum_{m\in\mathcal{M}} D_{1-p}  \left(\rho_{B^n}^m \| \rho_{B^n}\right) &\geq \frac{1}{ |\mathcal{M}|} 	\sum_{m\in\mathcal{M}} \frac{1}{\hat{p}}\log \left( \sum_{k\in\mathcal{K}} q(k)\left\| \Psi_{t}^{\otimes n}( \Pi_{B^n}^m ) \right\|_{\hat{p}, \rho_{B^n}^{x^n(m,k)} }^{\hat{p}} \right)  \\
	&\quad -  \frac{1}{|\mathcal{M}|}	\sum_{m\in\mathcal{M}}\log \Tr\left[ \rho_{B^n} \Psi_{t}^{\otimes n}( \Pi_{B^n}^m ) \right]. 
	\end{split} \label{eq:Fano3}
	\end{align}
	
	In the following, we lower bound the right-hand side of \eqref{eq:Fano3}.
	The superoperator $(\Psi_{t}^{\otimes n} - \Phi_{t,x^n(m,k)})$, where $\Phi_{t,x^n(m,k)}$ is the superoperator defined through \eqref{eq:semigroup},  is positivity-preserving for every $(m,k) \in \mathcal{M}\times \mathcal{K}$, since $\rho_B^x \leq \mathbb{I}_B$ for all $x\in\mathcal{X}$. (This can be proved by induction in $n$, as in the proof of \cite[Theorem 29]{BDR18}).
	Further, the non-commutative weighted $L_\mathsf{p}$-norm $\|\cdot \|_{\mathsf{p}, \rho}$ is monotone non-decreasing in its argument for every $\mathsf{p} \in \mathbb{R}\backslash \{0\}$ (which can be immediately verified from the definition \eqref{eq:wLp} by using Weyl's Monotonicity Theorem \cite[Corollary III.2.3]{Bha97}).
	Hence, for every $m\in\mathcal{M}$,
	\begin{align}
	\sum_{k\in\mathcal{K}} q(k) \left\| \Psi_{t}^{\otimes n}( \Pi_{B^n}^m ) \right\|_{\hat{p}, \rho_{B^n}^{x^n(m,k)} }^{\hat{p}} \leq
	\sum_{k\in\mathcal{K}} q(k)\left\| \Phi_{t,x^n(m,k)}^{\otimes n}( \Pi_{B^n}^m ) \right\|_{\hat{p}, \rho_{B^n}^{x^n(m,k)} }^{\hat{p}}. \label{eq:Fano4}
	\end{align}
	Then, we employ the Reverse Hypercontractivity, Lemma~\ref{lemm:RHC}, on the right-hand side of  ~\eqref{eq:Fano4} with $\mathsf{p} = \hat{p} \in (-1,0)$ and  any $\mathsf{q} = q \in (0,1)$ satisfying $t = \log \frac{\hat{p}-1}{q-1}$ to obtain 
	\begin{align}
	\left\| \Phi_{t,x^n(m,k)}(\Pi_{B^n}^m) \right\|_{\hat{p}, \rho_{B^n}^{x^n(m,k)} }
	&\geq \left\| \Pi_{B^n}^m \right\|_{q, \rho_{B^n}^{x^n(m,k)} } \\
	&= \left( \Tr\left[ \left( (\rho_{B^n}^{x^n(m,k)})^{\frac{1}{2q}} \Pi_{B^n}^m (\rho_{B^n}^{x^n(m,k)})^{\frac{1}{2q}} \right)^q  \right] \right)^{\frac{1}{q}} \\
	&\geq \left( \Tr\left[ \rho_{B^n}^{x^n(m,k)} (\Pi_{B^n}^m)^q  \right] \right)^{\frac{1}{q}} \label{eq:Fano5}\\		
	&\geq \left( \Tr\left[  \rho_{B^n}^{x^n(m,k)} \Pi_{B^n}^m    \right] \right)^{\frac{1}{q}}. \label{eq:Fano6}
	\end{align}
	Here, we used the Araki-Lieb-Thirring inequality, Lemma~\ref{lemm:ALT}, with $A = \Pi_{B^n}^m$, $B = (\rho_{B^n}^{x^n(m,k)})^{\frac1q}$, and $r = q \in(0,1)$
	to obtain the inequality~\eqref{eq:Fano5}. The inequality \eqref{eq:Fano6} holds because $0\leq \Pi_{B^n}^m \leq \mathbb{I}_{B^n}$, so that $(\Pi_{B^n}^m)^q \geq \Pi_{B^n}^m$ for $q\in(0,1)$. 
	From \eqref{eq:Fano4} and \eqref{eq:Fano6}, the first term on the right-hand side of  ~\eqref{eq:Fano3} is hence lower bounded by
	\begin{align}
	\frac{1}{ |\mathcal{M}|} 	\sum_{m\in\mathcal{M}} \frac{1}{\hat{p}}\log \left(   \sum_{k\in\mathcal{K}} q(k) \left( \Tr\left[  \rho_{B^n}^{x^n(m,k)} \Pi_{B^n}^m    \right] \right)^{\frac{\hat{p}}{q}} \right). \label{eq:Fano9} 
	\end{align}
	
	Next, we lower bound the second term on the right-hand side of  ~\eqref{eq:Fano3}.
	The concavity of the logarithm function implies that
	\begin{align}
	- \frac{1}{|\mathcal{M}|} 	\sum_{m\in\mathcal{M}} \log \Tr\left[ \rho_{B^n} \Psi_{t}^{\otimes n}( \Pi_{B^n}^m ) \right] &\geq - 	 \log \left(  \frac{1}{|\mathcal{M}|} \sum_{m\in\mathcal{M}} \Tr\left[ \rho_{B^n} \Psi_{t}^{\otimes n}( \Pi_{B^n}^m ) \right] \right) \\
	&= - \log \left( \frac{1}{|\mathcal{M}|}  \Tr\left[ \rho_{B^n} \Psi_{t}^{\otimes n} (\mathbb{I}_B^{\otimes n} ) \right]  \right),  \\
	&\geq \log |\mathcal{M}| - dnt, \label{eq:Fano10}
	\end{align}
	where the last inequality follows from the fact that\footnote{Note that the convexity of $h(u):= u^d$ for $d\geq 2$ implies that $(h(u)-h(1))/(u-1) \geq h'(1)$ for ever $u\geq1$. Hence, $\e^{dt} - 1 \geq d(\e^t-1)$ for every $t\geq 0$, and $\e^{-t} + d(1-\e^{-t}) \leq \e^{(d-1)t}$.}
	\begin{align}
	\Psi_t^{\otimes n}(\mathbb{I}_B^{\otimes n}) = (\e^{-t} + d (1-\e^{-t}))^n \mathbb{I}_B^{\otimes n} \leq \e^{(d-1)nt} \mathbb{I}_B^{\otimes n} \leq \e^{dnt} \mathbb{I}_B^{\otimes n}.
	\end{align}	
	
	Combining  ~\eqref{eq:Fano3}, \eqref{eq:Fano9}, and \eqref{eq:Fano10} yields
	\begin{align}
	&\frac{1}{|\mathcal{M}|} \sum_{m\in\mathcal{M}} D_{1-p}  \left(\rho_{B^n}^m \| \rho_{B^n}\right) \notag \\ &\geq 	\frac{1}{ |\mathcal{M}|} 	\sum_{m\in\mathcal{M}} \frac{1}{\hat{p}}\log \left(  \sum_{k\in\mathcal{K}} q(k)\left( \Tr\left[  \rho_{B^n}^{x^n(m,k)} \Pi_{B^n}^m    \right] \right)^{\frac{\hat{p}}{q}} \right) 
	+ \log |\mathcal{M}| - dnt. \label{eq:Fano11}
	\end{align}
	Next we take the limits $p\to 0$ and $\hat{p} \to 0$ (which in turn ensures that $q \to 1-\e^{-t}$) on both sides of the above inequality. Then the left-hand side of \eqref{eq:Fano11} becomes
	\begin{align}
	\lim_{p\to 0}
	\frac{1}{|\mathcal{M}|} \sum_{m\in\mathcal{M}} D_{1-p}  \left(\rho_{B^n}^m \| \rho_{B^n}\right) &= \frac{1}{|\mathcal{M}|} \sum_{m\in\mathcal{M}} D  \left(\rho_{B^n}^m \| \rho_{B^n}\right) \label{eq:Fano15} \\
	&= D( \rho_{MB^n}|| \rho_M \otimes \rho_{B^n}) = I(M; B^n)_\rho, \label{eq:Fano12}
	\end{align}
	where the equality \eqref{eq:Fano15} follows from the fact that the the quantum relative R\'enyi entropy $D_{1-p}$ converges to the quantum relative entropy $D$ as $p\to 0 $.
	
	\medskip
	
	On the other hand, the first term on the right-hand side of  ~\eqref{eq:Fano11} becomes
	\begin{align}
	&\lim_{ \hat{p} \to 0} 	\frac{1}{ |\mathcal{M}|} 	\sum_{m\in\mathcal{M}} \frac{1}{\hat{p}}\log \left(  \sum_{k\in\mathcal{K}} q(k) \left( \Tr\left[  \rho_{B^n}^{x^n(m,k)} \Pi_{B^n}^m    \right] \right)^{\frac{\hat{p}}{q}} \right) \notag \\ 
	&= \frac{1}{ |\mathcal{M}|} 	\sum_{(m,k) \in \mathcal{M}\times \mathcal{K}} q(k)\, \frac{1}{1-\e^{-t}} \log \Tr\left[  \rho_{B^n}^{x^n(m,k)} \Pi_{B^n}^m    \right] \label{eq:Fano18}\\
	&\geq \frac{1}{1-\e^{-t}} \log (1-\eps ) \label{eq:Fano19}\\
	&\geq -\left(1+\frac1t\right) \log \frac{1}{1-\eps }. \label{eq:Fano13}
	\end{align}
	In the above, equality \eqref{eq:Fano18} is due to L'H\^{o}spital's rule; inequalities \eqref{eq:Fano19} and \eqref{eq:Fano13} hold because of the assumption given in \eqref{eq:Fano_condition} and the fact that $\frac{1}{1-\e^{-t}} \leq 1+ \frac1t$.
	Finally, \eqref{eq:Fano11}, \eqref{eq:Fano12}, and \eqref{eq:Fano13}, together imply that
	\begin{align}
	\log |\mathcal{M}| \leq I(M;B^n)_\rho + dnt + \left(1+\frac1t\right) \log \frac{1}{1-\eps }. \label{eq:Fano14}
	\end{align}
The above bound~\eqref{eq:Fano14} can be shown to be optimized when 
	\begin{align}\label{tee}
	t = \sqrt{ \frac{ - \log (1-\eps) }{  d n } }, 
	\end{align}
which satisfies the requirement $t>0$ since $\eps \in (0,1)$. Substituting \eqref{tee} in \eqref{eq:Fano14} yields the desired result.

\end{proof}

\section{Second-Order Strong Converse Bound for a Classical-Quantum Degraded Broadcast Channel} \label{sec:sc_DBC}

Let us now revert to the c-q degraded broadcast channel $\mathscr{W}^{X \to BC}$ which was introduced in Section \ref{sec:c-q_DBC}, and is the focus of this paper.
In the following theorem we establish second-order (in $n$) upper bounds to rate pairs $(R_B, R_C)$ of any $(n, R_B, R_C, \eps)$ code for such a channel $\mathscr{W}^{X \to BC}$.
\begin{theo}
[Second-order strong converse bound for a c-q DBC] \label{theo:SC_DBC}
For a c-q DBC $\mathscr{W}^{X\to BC}$ as given in Definition~\ref{defn:channel}, any $(n, R_B, R_C, \eps)$ code satisfies
	\begin{align}
&	R_B \leq I(X;B|U)_\omega + 2 \sqrt{ \frac{d_B }{n} \log \frac{1}{1-\eps} } + \frac1n\log\frac{1}{1-\eps}; \\
&	R_C \leq I(U;C)_\omega + 2 \sqrt{ \frac{ d_C }{n} \log \frac{1}{1-\eps} } + \frac1n\log\frac{1}{1-\eps}, \label{eq:SC_DBC}
	\end{align}
	for some $\omega_{UXBC}$ of the form
	\begin{align} \label{eq:DBC_condition}
	\omega_{UXBC} &= \sum_{x \in \cX} p_X(x) |x \rangle \langle x| \otimes \rho_U^x \otimes \rho_{BC}^x
	\end{align}
	for some probability distribution $p$ on $\mathcal{X}$, and some collection of density matrices $\{\rho_U^x\}_{x\in\mathcal{X}}$.
\end{theo}

\begin{remark}
	Theorem~\ref{theo:SC_DBC} can be extended to the case in which Alice transmits common information (at a rate $R$, say), in addition to private information.
	It can be verified that if $(R, R_B, R_C)$ is an $\eps$-achievable rate triple, then $(0, R_B, R+ R_C)$ is also $\eps$-achievable. An intuitive way to see this is as follows: Bob can disregard the common information that he receives, while Charlie can consider the common information that he receives as part of his private information, without affecting the error probability of the protocol.	
	This was stated for the case $\eps=0$ for a c-q DBC in \cite{YHD11} and is well-known for the case of classical broadcast channels (see e.g.~\cite{GK11}).
	Due to this reason it suffices to incorporate the common information rate into the rate of private information transmission to Charlie.
	If the common information rate $R$ is non-zero, then the left-hand side of \eqref{eq:SC_DBC} should be read as $R+R_C$.
\end{remark}

\begin{proof}
	[Proof of Theorem~\ref{theo:SC_DBC}]
	
	Let 
	\begin{align}
	\rho_{MKX^nB^nC^n} = \frac{1}{|\mathcal{M}||\mathcal{K}|} \sum_{m\in \mathcal{M}} \sum_{k \in \mathcal{K}} |m\rangle\langle m| \otimes |k\rangle\langle k| \otimes |x^n(m,k)\rangle\langle x^n(m,k) |\otimes \rho_{B^n C^n}^{x^n(m,k)}.
	\end{align}
	Observe that
	\begin{align}
	\min\left\{ \Tr\left[ \rho_{B^n}^{x^n(m,k)} \Pi_{B^n}^{m}  \right] , \Tr\left[ \rho_{ C^n}^{x^n(m,k)} \Pi_{C^n}^{k} \right] \right\} \geq
	\Tr\left[ \rho_{B^n C^n}^{x^n(m,k)} \Pi_{B^n}^{m} \otimes \Pi_{C^n}^{k} \right] \geq 1-\eps
	\end{align}
	by definition of an $(n,R_B, R_C, \eps)$ code.
	Hence, the $(n,R_B, R_C, \eps)$-code satisfies the geometric average error criterion given by \eqref{eq:Fano_condition} (cf.~Remark \ref{remark:error})
	We then apply the second-order Fano-type inequality, Theorem~\ref{theo:Fano}, with the choice $q(k) = \frac{1}{|\mathcal{K}|}$ for every $k\in\mathcal{K}$	and $n R_B =  \log |\mathcal{M}|$, $n R_C = \log |\mathcal{K}|$
	to obtain the following upper bounds for the rate pair:
	\begin{align}
	\begin{split} \label{eq:DBC0}
	n R_B &\leq I(M;B^n)_\rho + 2 \sqrt{ n d_B  \log \frac{1}{1-\eps} } + \log\frac{1}{1-\eps};  \\
	n R_C &\leq I(K;C^n)_\rho + 2 \sqrt{ n d_C  \log \frac{1}{1-\eps} } + \log\frac{1}{1-\eps}. 
	\end{split}
	\end{align}
	
	To complete the proof, we need to find upper bounds on $I(M;B^n)_\rho$ and $I(K;C^n)_\rho$ in terms of single-letter entropic quantities.
	This was done by Yard \emph{et al.}~\cite[Theorem 2]{YHD11} following the same idea that was used by Gallager~\cite{Gal74} in the classical case, and which is
often referred to as \emph{identification of the auxiliary random variable} (see also \cite[Chapter 5.4]{GK11}).
	The upper bounds obtained are given by 
\begin{align}
\begin{split} \label{eq:upper}
I(M;B^n)_\rho &\leq n I(X;B|U)_{\omega}, \\
I(K; C^n)_\rho &\leq n I(U;C)_{\omega},
\end{split}
\end{align}
for some quantum state $\omega_{UXBC}$ of the form given in \eqref{eq:DBC_condition} of the statement of Theorem \ref{theo:SC_DBC}. For the sake of completeness, we include the proof in Appendix~\ref{proof:upper}.  This concludes the proof of Theorem~\ref{theo:SC_DBC}.
\end{proof}


Taking the limit $n\to \infty$, on both sides of the inequalities in Theorem~\ref{theo:SC_DBC} directly shows that the $\eps$-capacity region $\mathcal{C}_{\mathscr{W}}(\eps)$ is contained in $\mathcal{C}^\text{ent}_{\mathscr{W}}$ for all $\eps \in(0,1)$. This in turn demonstrates the strong converse property for the c-q DBC, stated in Corollary~\ref{coro:epsilon}.
In other words, for any sequence of codes with rate pair $(R_B, R_C) \not\in \mathcal{C}^{\operatorname{ent}}_{\mathscr{W}}$, transmission of private information from Alice to Bob and Charlie fails with certainty, no matter how many times the channel is used.

\begin{coro}
	[Strong Converse Property] \label{coro:epsilon}
	For a c-q DBC $\mathscr{W}^{X\to BC}$ as given in Definition~\ref{defn:channel}, the following holds:
	\begin{align}
	\mathcal{C}_{\mathscr{W}}(\eps) \subseteq \mathcal{C}_{\mathscr{W}}^{\textnormal{ent}}, \quad \forall \eps\in(0,1).
	\end{align}
\end{coro}

In fact, Theorem~\ref{theo:SC_DBC} yields a finite blocklength strong converse, namely, that the maximal error of any $(n, R_B, R_C)$-code converges to $1$ exponentially fast (in $n$) whenever $(R_B, R_C) \not\in \mathcal{C}^\text{ent}_{\mathscr{W}}$. This is stated in the following corollary.
\begin{coro}
	[Exponential Strong Converse] \label{coro:exponential}
	For a c-q DBC $\mathscr{W}^{X\to BC}$ as given in Definition~\ref{defn:channel} and any non-negative rate pair $(R_B, R_C) \not\in \mathcal{C}_{\mathscr{W}}^{\operatorname{ent}}$, the maximal error of any $(n, R_B, R_C)$ code $\left( \mathcal{E}_n, \mathcal{D}_n \right)$ satisfies
	\begin{align}
	p_{\max}\left( \mathcal{E}_n, \mathcal{D}_n \right) \geq 1 - \e^{-nf},
	\end{align}
	where 
	\begin{align}
	f = \left( \sqrt{ (\sqrt{ d_B } + \sqrt{ d_C })^2 + \eta } - \sqrt{d_B} - \sqrt{d_C} \right)^2 > 0
	\end{align}
	for some $\eta >0$ depending only on how far the rate pair $(R_B, R_C)$ is from the region $\mathcal{C}_\mathscr{W}^{\operatorname{ent}}$.
\end{coro}
\begin{proof}
	[Proof of Corollary~\ref{coro:exponential}]
	Let us first define the function
	\begin{align} 
	F(t) &:= \sup_{ { \rho    } } \left\{ I(X;B|U)_\rho:  I(U;C)_\rho \geq t \right\}, \quad \forall t\geq 0,
	\end{align}
	where the supremum is taken over all states $\rho\equiv \rho_{UXBC}$ of the form of \eqref{eq:DBC_condition}.
By the definition of $\mathcal{C}_\mathscr{W}^{\operatorname{ent}}$ given in \eqref{eq:C_W},
	$(R_B, R_C) \not\in \mathcal{C}_{\mathscr{W}}^{\operatorname{ent}}$ implies that
	\begin{align}
	R_B > F(R_C). \label{eq:exp2}
	\end{align}
	In Appendix~\ref{app:concavity}, we prove that $F(t)$ is a concave function in $t\geq 0$. Therefore, by the method of Lagrange multipliers, inequality \eqref{eq:exp2} can be further written as
	\begin{align}
		R_B > \inf_{\mu\geq 0} \sup_{ { \rho } } \left\{ I(X;B|U)_\rho + \mu I(U;C)_\rho - \mu R_C \right\}.
	\end{align}
	Hence, there must exist some $\mu^\star \in \mathbb{R}_{\geq 0}$ and $\gamma >0$ such that
	\begin{align}
	R_B + \mu^\star R_C \geq \sup_{ \rho  }  \left\{I(X;B|U)_\rho + \mu^\star I(U;C)_\rho \right\} + \gamma, \label{eq:exponential1}
	\end{align}
	
	On the other hand, Theorem~\ref{theo:SC_DBC} guarantees that any $(n, R_B, R_C)$ code $\left( \mathcal{E}_n, \mathcal{D}_n \right)$ with $p_{\max} \left( \mathcal{E}_n, \mathcal{D}_n \right)\leq \eps \in (0,1)$	
	 satisfies
	\begin{align}
	R_B \leq I(X;B|U)_\omega + 2 \sqrt{ \frac{d_B }{n} \log \frac{1}{1-\eps} } + \frac1n\log\frac{1}{1-\eps}, \\
	R_C \leq I(U;C)_\omega + 2 \sqrt{ \frac{d_C }{n} \log \frac{1}{1-\eps} } + \frac1n\log\frac{1}{1-\eps}
	\end{align}
	for some $\omega_{UXBC}$ of the form \eqref{eq:DBC_condition}.
	Defining $x_n^2 := \log \frac{1}{1-\eps}$, then we have
	\begin{align}
	R_B + \mu^\star R_C &\leq I(X;B|U)_\omega + \mu^\star I(U;C)_\omega + 2\frac{(1+\mu^\star)}{\sqrt{n}} (\sqrt{d_B} + \sqrt{d_B}) x_n +	\frac{(1+\mu^\star)}{n} x_n^2 \\
	&\leq  \sup_{ \rho   }  \left\{I(X;B|U)_\rho + \mu^\star I(U;C)_\rho \right\} + 2\frac{(1+\mu^\star)}{\sqrt{n}} (\sqrt{d_B} + \sqrt{d_C}) x_n +	\frac{(1+\mu^\star)}{n} x_n^2. \label{eq:exponential2}
	\end{align}
	Combining \eqref{eq:exponential1} and \eqref{eq:exponential2} gives
	\begin{align}
	(1+\mu^\star) x_n^2 + 2(1+\mu^\star)(\sqrt{n d_B} + \sqrt{n d_C}) x_n - n \gamma \geq 0.
	\end{align}
	Solving this and choosing $\eta = \frac{\gamma}{1+\mu^\star} > 0$ concludes the proof of the corollary.	
\end{proof}

\section*{Acknowledgements}
HC was supported by the Cambridge University Fellowship and the Ministry of Science and Technology Overseas Project for Post Graduate Research (Taiwan) under Grant 108-2917-I-564-042.
CR is supported by the TUM University Foundation Fellowship.
We thank Jingbo Liu for helpful discussions.

\appendix
\section{Proof of (\ref{eq:upper})  } \label{proof:upper}
For each $i \in \{1,\ldots, n\}$, we introduce an auxiliary composite quantum system $U_i = (K,B^{i-1})$.
We upper bound the first term in ~\eqref{eq:upper} as follows:
\begin{align}
I(M;B^n)_\rho &\leq I(M; KB^n)_\rho \label{eq:DBC8} \\
&= I(M; KB^n)_\rho - I(M;K)_\rho \label{eq:DBC1}\\
&= I(M; B^n | K)_\rho \label{eq:DBC2}\\
&= \sum_{i=1}^n I(M; B_i | K, B^{i-1})_\rho \label{eq:DBC3}\\
&= \sum_{i=1}^n I(M; B_i | U_i )_\rho \\
&\leq \sum_{i=1}^n I(M; B_i | U_i )_\rho + I(X_i; Y_i| M U_i)_\rho \label{eq:DBC4}\\
&= \sum_{i=1}^n I(X_i, M; B_i | U_i )_\rho \label{eq:DBC5} \\
&= \sum_{i=1}^n I(X_i ; B_i| U_i )_\rho + I(M; B_i | X_i U_i )_\rho \label{eq:DBC6}\\
&= \sum_{i=1}^n I(X_i ; B_i | U_i )_\rho. \label{eq:DBC7}
\end{align}
Here, inequality \eqref{eq:DBC8} is due to monotonicity of the mutual information with respect to the partial trace. Identity~\eqref{eq:DBC1} is because $M$ and $K$ are uncorrelated.
Equalities \eqref{eq:DBC2}, \eqref{eq:DBC3}, \eqref{eq:DBC5}, and \eqref{eq:DBC6} follow from the chain rule of quantum mutual information:	
$I(A^n:C |B)_\rho = \sum_{i=1}^n I(A_i;C|B, A^{i-1})_\rho$ and 
$I(A; C^n| B)_\rho = \sum_{i=1}^n I(A;C_i|B, C^{i-1})_\rho$.
Inequality \eqref{eq:DBC4} is due to the non-negativity of the conditional quantum mutual information.
The last line \eqref{eq:DBC7} holds because of the Markov chain:
$M - (K,X_i, B^{i-1}) - B_i$.
To see this, the right quantum system $B_i$ can be produced by knowing the value of $X_i$.

Next, we consider the second term in  ~\eqref{eq:upper}:
\begin{align}
I(K; C^n)_\rho &= \sum_{i=1}^n I(K; C_i| C^{i-1} )_\rho \label{eq:DBC10} \\
&=  \sum_{i=1}^n H(C_i| C^{i-1})_\rho - H(C_i| K C^{i-1})_\rho \\
&\leq \sum_{i=1}^n H(C_i)_\rho - H(C_i| K C^{i-1})_\rho \label{eq:DBC11} \\
&\leq \sum_{i=1}^n H(C_i)_\rho - H(C_i| K B^{i-1})_\rho \label{eq:DBC12} \\
&= \sum_{i=1}^n I(K, B^{i-1}; C_i)_\rho. \label{eq:DBC9}
\end{align}
Here, equalities \eqref{eq:DBC10} and \eqref{eq:DBC9} are again by the chain rule. 
Inequality \eqref{eq:DBC11} is because conditioning reduces entropies.
Inequality \eqref{eq:DBC12} follows from the data processing with respect to the tensor product of the degrading quantum operation $\mathcal{N}^{B\to C}$.

Now, we introduce a \emph{time-sharing} random variable $T$ that is uniform on $\{1,\ldots, n\}$ and independent of other systems.
Identify $U = (T,K,B^{T-1})$, which clearly satisfies  ~\eqref{eq:DBC_condition}.
We have the following bounds of  ~\eqref{eq:DBC7} and \eqref{eq:DBC9}, respectively:
\begin{align}
\sum_{i=1}^n I(X_i, ; B_i | U_i )_\rho &= n I(X_T; B_T| TK B^{T-1} )_{T\otimes \rho} \\
&= n I(X;B|U)_{\omega},
\end{align}
and
\begin{align}
\sum_{i=1}^n I(K B^{i-1}; C_i)_\rho &= n I(K, B^{T-1}; C_T | T)_{T\otimes \rho} \\
&\leq n \left[  I(K B^{T-1}; C_T | T)_{T\otimes \rho} + I(T; C_T)_{T\otimes \rho} \right] \\
&= n I(I, K B^{T-1}; C_T)_{T\otimes \rho} \\
&= n I(U;C)_{\omega}.
\end{align}
\qed


\section{A Concavity Property} \label{app:concavity}

We define the following function:
\begin{align} 
F(t) &:= \sup_{ { \rho \in \Sigma(\mathscr{W})   } } \left\{ I(X;B|U)_\rho:  I(U;C)_\rho \geq t \right\}, \quad \forall t\geq 0; \label{eq:rate_function} \\
\Sigma(\mathscr{W}) &:= \left\{
\rho_{UXBC} = \bigoplus_{ x \in\mathcal{X}  }\, p(x) \rho_U^x\otimes \rho_{BC}^x: p \text{ is a probability distribution on } \mathcal{X}, \text{and } \left\{\rho_U^x \right\}_{x\in\mathcal{X}} \subset \mathcal{D}(\mathcal{H}_U)
\right\}. \label{eq:set}
\end{align}

The following concavity of the function $F(t)$ can be proved by following similar idea of Ahlswede and K\"orner \cite{AK75}. For completeness, we provide a proof here.
\begin{theo}
	\label{theo:concavity}
	The function $F(t)$ defined in ~\eqref{eq:rate_function} is concave for all $t\geq 0$.
	Moreover,
	\begin{align} \label{eq:infimum}
	F(t) = \inf_{ \mu \geq 0 } \sup_{\rho\in\Sigma(\mathscr{W}) } \left\{
	 I(X;B|U)_\rho  + \mu I(U;C)_\rho - \mu t
	\right\}.
	\end{align}
\end{theo}
\begin{proof}
	We aim to prove
	\begin{align}
	F(\lambda t_0 + (1-\lambda) t_1) \geq \lambda F(t_0) + (1-\lambda) F(t_1),
	\end{align}
	for all $\lambda \in [0,1]$ and $t_0, t_1 \geq 0$.
	For every $\gamma>0$, let $\rho_0, \rho_1 \in \Sigma(\mathscr{W})$ such that
	$I(X;B|U)_{\rho_i} \geq F(t_i) - \gamma $ and $ I(U;C)_{\rho_i} \geq t_i$ for $i \in \{0,1\}$.
	
	Now, we introduce a new Bernoulli random variable $V$ with $\Pr\{V=0\} = \lambda$ and $ \Pr\{V = 1\} = (1-\lambda)$ such that
	\begin{align}
	\rho_{VUBC} := \lambda |0\rangle\langle 0| \otimes \rho_0 + (1-\lambda) |1\rangle\langle 1| \otimes \rho_1 \in \Sigma(\mathscr{W}).
	\end{align}
	From the choice of $\rho_0, \rho_1$, and $\rho$, we have
	\begin{align}
	\lambda F(t_0) + (1-\lambda) F(t_1) - \gamma
	&\leq \lambda I(X;B|U)_{\rho_0} + (1-\lambda) I(X;B|U)_{\rho_1} \\
	&= I(X;B|UV)_\rho.
	\end{align}
	On the other hand, using the chain rule and non-negativity of quantum mutual information, we have
	\begin{align}
	\lambda t_0 + (1-\lambda) t_1 &\leq
	\lambda I(U;C)_{\rho_0} + (1-\lambda) I(U;C)_{\rho_1} \\
	&= I(U;C|V)_{\rho} \\
	&= I(VU; C)_\rho - I(V;C)_\rho \\
	&\leq I(VU;C)_\rho.
	\end{align}
	This means that $\rho$ satisfies the constraint of  in the definition of $F(\lambda t_0 + (1-\lambda) t_1)$. Therefore,
	\begin{align}
	F(\lambda t_0 + (1-\lambda) t_1) &\geq I(X;B|VU)_\rho \\
	&\geq \lambda F(t_0) + (1-\lambda) F(t_1) - \gamma.
	\end{align}
	Since this holds for every $\gamma>0$, we conclude the proof by letting $\gamma\to 0$.
	The second assertion in \eqref{eq:infimum} follows from the method of Lagrange multipliers and the concavity of $F(t)$.
\end{proof}



\begin{thebibliography}{10}
	\providecommand{\url}[1]{#1}
	\csname url@samestyle\endcsname
	\providecommand{\newblock}{\relax}
	\providecommand{\bibinfo}[2]{#2}
	\providecommand{\BIBentrySTDinterwordspacing}{\spaceskip=0pt\relax}
	\providecommand{\BIBentryALTinterwordstretchfactor}{4}
	\providecommand{\BIBentryALTinterwordspacing}{\spaceskip=\fontdimen2\font plus
		\BIBentryALTinterwordstretchfactor\fontdimen3\font minus
		\fontdimen4\font\relax}
	\providecommand{\BIBforeignlanguage}[2]{{%
			\expandafter\ifx\csname l@#1\endcsname\relax
			\typeout{** WARNING: IEEEtran.bst: No hyphenation pattern has been}%
			\typeout{** loaded for the language `#1'. Using the pattern for}%
			\typeout{** the default language instead.}%
			\else
			\language=\csname l@#1\endcsname
			\fi
			#2}}
	\providecommand{\BIBdecl}{\relax}
	\BIBdecl
	
	\normalsize
	
	\bibitem{Cov72}
	T.~Cover, ``Broadcast channels,''
	\href{http://dx.doi.org/10.1109/tit.1972.1054727}{\emph{{IEEE} Transactions
			on Information Theory}},
	\href{http://dx.doi.org/10.1109/tit.1972.1054727}{vol.~18},
	\href{http://dx.doi.org/10.1109/tit.1972.1054727}{no.~1},
	\href{http://dx.doi.org/10.1109/tit.1972.1054727}{pp. 2--14},
	\href{http://dx.doi.org/10.1109/tit.1972.1054727}{jan 1972}.
	
	\bibitem{Ber73}
	P.~Bergmans, ``Random coding theorem for broadcast channels with degraded
	components,'' \href{http://dx.doi.org/10.1109/tit.1973.1054980}{\emph{{IEEE}
			Transactions on Information Theory}},
	\href{http://dx.doi.org/10.1109/tit.1973.1054980}{vol.~19},
	\href{http://dx.doi.org/10.1109/tit.1973.1054980}{no.~2},
	\href{http://dx.doi.org/10.1109/tit.1973.1054980}{pp. 197--207},
	\href{http://dx.doi.org/10.1109/tit.1973.1054980}{mar 1973}.
	
	\bibitem{Wyn73}
	A.~Wyner, ``A theorem on the entropy of certain binary sequences and
	applications--{II},''
	\href{http://dx.doi.org/10.1109/tit.1973.1055108}{\emph{{IEEE} Transactions
			on Information Theory}},
	\href{http://dx.doi.org/10.1109/tit.1973.1055108}{vol.~19},
	\href{http://dx.doi.org/10.1109/tit.1973.1055108}{no.~6},
	\href{http://dx.doi.org/10.1109/tit.1973.1055108}{pp. 772--777},
	\href{http://dx.doi.org/10.1109/tit.1973.1055108}{nov 1973}.
	
	\bibitem{Ber74}
	P.~Bergmans, ``A simple converse for broadcast channels with additive white
	gaussian noise (corresp.),''
	\href{http://dx.doi.org/10.1109/tit.1974.1055184}{\emph{{IEEE} Transactions
			on Information Theory}},
	\href{http://dx.doi.org/10.1109/tit.1974.1055184}{vol.~20},
	\href{http://dx.doi.org/10.1109/tit.1974.1055184}{no.~2},
	\href{http://dx.doi.org/10.1109/tit.1974.1055184}{pp. 279--280},
	\href{http://dx.doi.org/10.1109/tit.1974.1055184}{mar 1974}.
	
	\bibitem{Ber77}
	M.~S. Berger, \emph{Nonlinearity and Functional Analysis}.\hskip 1em plus 0.5em
	minus 0.4em\relax Academic Press, 1977.
	
	\bibitem{AGK76}
	R.~Ahlswede, P.~G{\'a}cs, and J.~K{\"o}rner, ``Bounds on conditional
	probabilities with applications in multi-user communication,''
	\href{http://dx.doi.org/10.1007/bf00535682}{\emph{Zeitschrift f{\"u}r
			Wahrscheinlichkeitstheorie und Verwandte Gebiete}},
	\href{http://dx.doi.org/10.1007/bf00535682}{vol.~34},
	\href{http://dx.doi.org/10.1007/bf00535682}{no.~2},
	\href{http://dx.doi.org/10.1007/bf00535682}{pp. 157--177},
	\href{http://dx.doi.org/10.1007/bf00535682}{1976}.
	
	\bibitem{Gal74}
	\BIBentryALTinterwordspacing
	R.~G. Gallager, ``Capacity and coding for degraded broadcast channels,''
	\emph{Probl. Peredachi Inf.}, vol.~10, no.~3, pp. 3--14, 1974. [Online].
	Available: \url{http://mi.mathnet.ru/eng/ppi1036}
	\BIBentrySTDinterwordspacing
	
	\bibitem{vdM75}
	E.~van~der Meulen, ``Random coding theorems for the general discrete memoryless
	broadcast channel,''
	\href{http://dx.doi.org/10.1109/tit.1975.1055347}{\emph{{IEEE} Transactions
			on Information Theory}},
	\href{http://dx.doi.org/10.1109/tit.1975.1055347}{vol.~21},
	\href{http://dx.doi.org/10.1109/tit.1975.1055347}{no.~2},
	\href{http://dx.doi.org/10.1109/tit.1975.1055347}{pp. 180--190},
	\href{http://dx.doi.org/10.1109/tit.1975.1055347}{mar 1975}.
	
	\bibitem{Cov75}
	T.~Cover, ``An achievable rate region for the broadcast channel,''
	\href{http://dx.doi.org/10.1109/tit.1975.1055418}{\emph{{IEEE} Transactions
			on Information Theory}},
	\href{http://dx.doi.org/10.1109/tit.1975.1055418}{vol.~21},
	\href{http://dx.doi.org/10.1109/tit.1975.1055418}{no.~4},
	\href{http://dx.doi.org/10.1109/tit.1975.1055418}{pp. 399--404},
	\href{http://dx.doi.org/10.1109/tit.1975.1055418}{jul 1975}.
	
	\bibitem{KM77}
	J.~Korner and K.~Marton, ``General broadcast channels with degraded message
	sets,'' \href{http://dx.doi.org/10.1109/tit.1977.1055655}{\emph{{IEEE}
			Transactions on Information Theory}},
	\href{http://dx.doi.org/10.1109/tit.1977.1055655}{vol.~23},
	\href{http://dx.doi.org/10.1109/tit.1977.1055655}{no.~1},
	\href{http://dx.doi.org/10.1109/tit.1977.1055655}{pp. 60--64},
	\href{http://dx.doi.org/10.1109/tit.1977.1055655}{jan 1977}.
	
	\bibitem{vdM77}
	E.~van~der Meulen, ``A survey of multi-way channels in information theory:
	1961-1976,'' \href{http://dx.doi.org/10.1109/tit.1977.1055652}{\emph{{IEEE}
			Transactions on Information Theory}},
	\href{http://dx.doi.org/10.1109/tit.1977.1055652}{vol.~23},
	\href{http://dx.doi.org/10.1109/tit.1977.1055652}{no.~1},
	\href{http://dx.doi.org/10.1109/tit.1977.1055652}{pp. 1--37},
	\href{http://dx.doi.org/10.1109/tit.1977.1055652}{jan 1977}.
	
	\bibitem{Sat78}
	H.~Sato, ``An outer bound to the capacity region of broadcast channels
	(corresp.),'' \href{http://dx.doi.org/10.1109/tit.1978.1055883}{\emph{{IEEE}
			Transactions on Information Theory}},
	\href{http://dx.doi.org/10.1109/tit.1978.1055883}{vol.~24},
	\href{http://dx.doi.org/10.1109/tit.1978.1055883}{no.~3},
	\href{http://dx.doi.org/10.1109/tit.1978.1055883}{pp. 374--377},
	\href{http://dx.doi.org/10.1109/tit.1978.1055883}{may 1978}.
	
	\bibitem{Mar79}
	K.~Marton, ``A coding theorem for the discrete memoryless broadcast channel,''
	\href{http://dx.doi.org/10.1109/tit.1979.1056046}{\emph{{IEEE} Transactions
			on Information Theory}},
	\href{http://dx.doi.org/10.1109/tit.1979.1056046}{vol.~25},
	\href{http://dx.doi.org/10.1109/tit.1979.1056046}{no.~3},
	\href{http://dx.doi.org/10.1109/tit.1979.1056046}{pp. 306--311},
	\href{http://dx.doi.org/10.1109/tit.1979.1056046}{may 1979}.
	
	\bibitem{Gam79}
	A.~Gamal, ``The capacity of a class of broadcast channels,''
	\href{http://dx.doi.org/10.1109/tit.1979.1056029}{\emph{{IEEE} Transactions
			on Information Theory}},
	\href{http://dx.doi.org/10.1109/tit.1979.1056029}{vol.~25},
	\href{http://dx.doi.org/10.1109/tit.1979.1056029}{no.~2},
	\href{http://dx.doi.org/10.1109/tit.1979.1056029}{pp. 166--169},
	\href{http://dx.doi.org/10.1109/tit.1979.1056029}{mar 1979}.
	
	\bibitem{GvdM81}
	A.~E. Gamal and E.~van~der Meulen, ``A proof of {Marton's} coding theorem for
	the discrete memoryless broadcast channel (corresp.),''
	\href{http://dx.doi.org/10.1109/tit.1979.1056046}{\emph{{IEEE} Transactions
			on Information Theory}},
	\href{http://dx.doi.org/10.1109/tit.1979.1056046}{vol.~27},
	\href{http://dx.doi.org/10.1109/tit.1979.1056046}{no.~1},
	\href{http://dx.doi.org/10.1109/tit.1979.1056046}{pp. 120--122},
	\href{http://dx.doi.org/10.1109/tit.1979.1056046}{Jan 1981}.
	
	\bibitem{Cov98}
	T.~Cover, ``Comments on broadcast channels,''
	\href{http://dx.doi.org/10.1109/18.720547}{\emph{{IEEE} Transactions on
			Information Theory}}, \href{http://dx.doi.org/10.1109/18.720547}{vol.~44},
	\href{http://dx.doi.org/10.1109/18.720547}{no.~6},
	\href{http://dx.doi.org/10.1109/18.720547}{pp. 2524--2530},
	\href{http://dx.doi.org/10.1109/18.720547}{oct 1998}.
	
	\bibitem{Nai10}
	C.~Nair, ``Capacity regions of two new classes of two-receiver broadcast
	channels,'' \href{http://dx.doi.org/10.1109/tit.2010.2054310}{\emph{{IEEE}
			Transactions on Information Theory}},
	\href{http://dx.doi.org/10.1109/tit.2010.2054310}{vol.~56},
	\href{http://dx.doi.org/10.1109/tit.2010.2054310}{no.~9},
	\href{http://dx.doi.org/10.1109/tit.2010.2054310}{pp. 4207--4214},
	\href{http://dx.doi.org/10.1109/tit.2010.2054310}{sep 2010}.
	
	\bibitem{Ooh15a}
	Y.~Oohama, ``Strong converse exponent for degraded broadcast channels at rates
	outside the capacity region,'' in \emph{2015 {IEEE} International Symposium
		on Information Theory ({ISIT})}.\hskip 1em plus 0.5em minus 0.4em\relax
	{IEEE}, jun 2015.
	
	\bibitem{Ooh15b}
	------, ``Strong converse theorems for degraded broadcast channels with
	feedback,'' in \emph{2015 {IEEE} International Symposium on Information
		Theory ({ISIT})}.\hskip 1em plus 0.5em minus 0.4em\relax {IEEE}, jun 2015.
	
	\bibitem{Ooh16}
	Y.~{Oohama}, ``Exponent function for asymmetric broadcast channels at rates
	outside the capacity region,'' in \emph{2016 International Symposium on
		Information Theory and Its Applications (ISITA)}, Oct 2016, pp. 537--541.
	
	\bibitem{GK11}
	A.~E. Gamal and Y.-H. Kim, \emph{Network Information Theory}.\hskip 1em plus
	0.5em minus 0.4em\relax Cambridge University Press, dec 2011.
	
	\bibitem{AK75}
	R.~Ahlswede and J.~K{\"o}rner, ``Source coding with side information and a
	converse for degraded broadcast channels,''
	\href{http://dx.doi.org/10.1109/tit.1975.1055469}{\emph{{IEEE} Transactions
			on Information Theory}},
	\href{http://dx.doi.org/10.1109/tit.1975.1055469}{vol.~21},
	\href{http://dx.doi.org/10.1109/tit.1975.1055469}{no.~6},
	\href{http://dx.doi.org/10.1109/tit.1975.1055469}{pp. 629--637},
	\href{http://dx.doi.org/10.1109/tit.1975.1055469}{Nov 1975}.
	
	\bibitem{YHD11}
	J.~Yard, P.~Hayden, and I.~Devetak, ``Quantum broadcast channels,''
	\href{http://dx.doi.org/10.1109/tit.2011.2165811}{\emph{{IEEE} Transactions
			on Information Theory}},
	\href{http://dx.doi.org/10.1109/tit.2011.2165811}{vol.~57},
	\href{http://dx.doi.org/10.1109/tit.2011.2165811}{no.~10},
	\href{http://dx.doi.org/10.1109/tit.2011.2165811}{pp. 7147--7162},
	\href{http://dx.doi.org/10.1109/tit.2011.2165811}{oct 2011}.
	
	\bibitem{SW12}
	N.~Sharma and N.~A. Warsi, ``Fundamental bound on the reliability of quantum
	information transmission,''
	\href{http://dx.doi.org/10.1103/physrevlett.110.080501}{\emph{Physical Review
			Letters}}, \href{http://dx.doi.org/10.1103/physrevlett.110.080501}{vol. 110},
	\href{http://dx.doi.org/10.1103/physrevlett.110.080501}{no.~8},
	\href{http://dx.doi.org/10.1103/physrevlett.110.080501}{Feb 2013}.
	
	\bibitem{Pet86}
	D.~Petz, ``Quasi-entropies for finite quantum systems,''
	\href{http://dx.doi.org/10.1016/0034-4877(86)90067-4}{\emph{Reports on
			Mathematical Physics}},
	\href{http://dx.doi.org/10.1016/0034-4877(86)90067-4}{vol.~23},
	\href{http://dx.doi.org/10.1016/0034-4877(86)90067-4}{no.~1},
	\href{http://dx.doi.org/10.1016/0034-4877(86)90067-4}{pp. 57--65},
	\href{http://dx.doi.org/10.1016/0034-4877(86)90067-4}{Feb 1986}.
	
	\bibitem{Tom16}
	M.~Tomamichel, \emph{Quantum Information Processing with Finite
		Resources}.\hskip 1em plus 0.5em minus 0.4em\relax Springer International
	Publishing, 2016.
	
	\bibitem{Ume62}
	H.~Umegaki, ``Conditional expectation in an operator algebra. {IV}. entropy and
	information,'' \href{http://dx.doi.org/10.2996/kmj/1138844604}{\emph{Kodai
			Mathematical Seminar Reports}},
	\href{http://dx.doi.org/10.2996/kmj/1138844604}{vol.~14},
	\href{http://dx.doi.org/10.2996/kmj/1138844604}{no.~2},
	\href{http://dx.doi.org/10.2996/kmj/1138844604}{pp. 59--85},
	\href{http://dx.doi.org/10.2996/kmj/1138844604}{1962}.
	
	\bibitem{Ari76}
	S.~Arimoto, ``Computation of random coding exponent functions,''
	\href{http://dx.doi.org/10.1109/tit.1976.1055640}{\emph{IEEE Transactions on
			Information Theory}},
	\href{http://dx.doi.org/10.1109/tit.1976.1055640}{vol.~22},
	\href{http://dx.doi.org/10.1109/tit.1976.1055640}{no.~6},
	\href{http://dx.doi.org/10.1109/tit.1976.1055640}{pp. 665--671},
	\href{http://dx.doi.org/10.1109/tit.1976.1055640}{Nov 1976}.
	
	\bibitem{LT76}
	E.~H. Lieb and W.~E. Thirring, ``Inequalities for the moments of the
	eigenvalues of the {Schrodinger} {Hamiltonian} and their relation to
	{Sobolev} inequalities.''\hskip 1em plus 0.5em minus 0.4em\relax Walter de
	Gruyter {GmbH}.
	
	\bibitem{BDR18}
	S.~Beigi, N.~Datta, and C.~Rouz{\'e}, ``Quantum reverse hypercontractivity: its
	tensorization and application to strong converses,''
	\href{http://arxiv.org/abs/1804.10100}{\texttt{arXiv:1804.10100 [quant-ph]}}.		
	
	\bibitem{mossel2013reverse}
	E.~Mossel, K.~Oleszkiewicz, and A.~Sen, ``On reverse hypercontractivity,''
	\emph{Geometric and Functional Analysis}, vol.~23, no.~3, pp. 1062--1097,
	2013.
	
	\bibitem{CKMT15}
	T.~Cubitt, M.~Kastoryano, A.~Montanaro, and K.~Temme, ``Quantum reverse
	hypercontractivity,''
	\href{http://dx.doi.org/10.1063/1.4933219}{\emph{Journal of Mathematical
			Physics}}, \href{http://dx.doi.org/10.1063/1.4933219}{vol.~56},
	\href{http://dx.doi.org/10.1063/1.4933219}{no.~10},
	\href{http://dx.doi.org/10.1063/1.4933219}{p. 102204},
	\href{http://dx.doi.org/10.1063/1.4933219}{oct 2015}.
	
	\bibitem{Wil90}
	F.~Willems, ``\BIBforeignlanguage{English}{The maximal-error and average-error
		capacity region of the broadcast channel are identical : A direct proof},''
	\emph{\BIBforeignlanguage{English}{Problems of Control and Information
			Theory}}, vol.~19, no.~4, pp. 339--347, 12 1990.
	
	\bibitem{LHV18}
	J.~Liu, R.~van Handel, and S.~Verd{\'u}, ``Second-order converses via reverse hypercontractivity,''
	\href{http://arxiv.org/abs/1812.10129}{\texttt{arXiv:1812.10129 [cs.IT]}}.		
	
	\bibitem{CDR19a}
	H.-C. Cheng, N.~Datta, , and C.~Rouz\'{e}, ``Strong converse bounds in quantum
	network information theory: distributed hypothesis testing and source
	coding,'' \emph{(in preparation)}, 2019.
	
	\bibitem{Fan61}
	R.~M. Fano, \emph{Transmission of Information, A Statistical Theory of
		Communications}.\hskip 1em plus 0.5em minus 0.4em\relax The MIT Press, 1961.
	
	\bibitem{LCV17}
	J.~Liu, P.~Cuff, and S.~Verd{\'u}, ``On {$\alpha$}-decodability and
	{$\alpha$}-likelihood decoder,'' in \emph{2017 55th Annual Allerton
		Conference on Communication, Control, and Computing (Allerton)}.\hskip 1em
	plus 0.5em minus 0.4em\relax {IEEE}, oct 2017.
	
	\bibitem{CK11}
	I.~Csisz{\'a}r and J.~K{\"o}rner, \emph{Information Theory: Coding Theorems for
		Discrete Memoryless Systems}.\hskip 1em plus 0.5em minus 0.4em\relax
	Cambridge University Press ({CUP}), 2011.
	
	\bibitem{Ahl78}
	R.~Ahlswede, ``Elimination of correlation in random codes for arbitrarily
	varying channels,''
	\href{http://dx.doi.org/10.1007/bf00533053}{\emph{Zeitschrift f{\"u}r
			Wahrscheinlichkeitstheorie und Verwandte Gebiete}},
	\href{http://dx.doi.org/10.1007/bf00533053}{vol.~44},
	\href{http://dx.doi.org/10.1007/bf00533053}{no.~2},
	\href{http://dx.doi.org/10.1007/bf00533053}{pp. 159--175},
	\href{http://dx.doi.org/10.1007/bf00533053}{1978}.
	
	\bibitem{Don86}
	M.~J. Donald, ``On the relative entropy,''
	\href{http://dx.doi.org/10.1007/bf01212339}{\emph{Communications in
			Mathematical Physics}}, \href{http://dx.doi.org/10.1007/bf01212339}{vol.
		105}, \href{http://dx.doi.org/10.1007/bf01212339}{no.~1},
	\href{http://dx.doi.org/10.1007/bf01212339}{pp. 13--34},
	\href{http://dx.doi.org/10.1007/bf01212339}{mar 1986}.
	
	\bibitem{HP91}
	F.~Hiai and D.~Petz, ``The proper formula for relative entropy and its
	asymptotics in quantum probability,''
	\href{http://dx.doi.org/10.1007/bf02100287}{\emph{Communications in
			Mathematical Physics}}, \href{http://dx.doi.org/10.1007/bf02100287}{vol.
		143}, \href{http://dx.doi.org/10.1007/bf02100287}{no.~1},
	\href{http://dx.doi.org/10.1007/bf02100287}{pp. 99--114},
	\href{http://dx.doi.org/10.1007/bf02100287}{Dec 1991}.
	
	\bibitem{Pet86b}
	D.~Petz, ``Sufficient subalgebras and the relative entropy of states of a {von
		Neumann} algebra,''
	\href{http://dx.doi.org/10.1007/bf01212345}{\emph{Communications in
			Mathematical Physics}}, \href{http://dx.doi.org/10.1007/bf01212345}{vol.
		105}, \href{http://dx.doi.org/10.1007/bf01212345}{no.~1},
	\href{http://dx.doi.org/10.1007/bf01212345}{pp. 123--131},
	\href{http://dx.doi.org/10.1007/bf01212345}{mar 1986}.
	
	\bibitem{BFT17}
	M.~Berta, O.~Fawzi, and M.~Tomamichel, ``On variational expressions for quantum
	relative entropies,''
	\href{http://dx.doi.org/10.1007/s11005-017-0990-7}{\emph{Letters in
			Mathematical Physics}},
	\href{http://dx.doi.org/10.1007/s11005-017-0990-7}{vol. 107},
	\href{http://dx.doi.org/10.1007/s11005-017-0990-7}{no.~12},
	\href{http://dx.doi.org/10.1007/s11005-017-0990-7}{pp. 2239--2265},
	\href{http://dx.doi.org/10.1007/s11005-017-0990-7}{sep 2017}.
	
	\bibitem{Bha97}
	R.~Bhatia, \emph{Matrix Analysis}.\hskip 1em plus 0.5em minus 0.4em\relax
	Springer New York, 1997.
	
\end{thebibliography}


\end{document}